\documentclass[reqno,11pt]{amsart}
\usepackage[utf8]{inputenc}
\usepackage{graphicx}
\usepackage{slashed}
\usepackage{amscd}
\usepackage{amssymb}
\usepackage{esint}
\usepackage{enumitem}
\usepackage{comment}
\usepackage{amsmath}
\usepackage{dsfont}
\usepackage[mathscr]{eucal}
\usepackage{pstricks}
\usepackage{pst-all}

\numberwithin{equation}{section}
\allowdisplaybreaks[1]
\definecolor{labelkey}{gray}{.65}

\title[Baryogenesis in Minkowski Spacetime]{Baryogenesis in Minkowski Spacetime}

\author[F.\ Finster]{Felix Finster}
\address{Fakult\"at f\"ur Mathematik \\ Universit\"at Regensburg \\ D-93040 Regensburg \\ Germany}
\email{finster@ur.de}

\author[M. van den Beld-Serrano]{Marco van den Beld-Serrano \\ \\ August / November 2024}
\address{Fakult\"at f\"ur Mathematik \\ Universit\"at Regensburg \\ D-93040 Regensburg \\ Germany}
\email{marco.van-den-beld-serrano@ur.de}

\newtheorem{Def}{Definition}[section]
\newtheorem{Thm}[Def]{Theorem}

\newtheorem{Lemma}[Def]{Lemma}
\newtheorem{Remark}[Def]{Remark}
\newtheorem{MotEx}[Def]{Motivating Example}

\newtheorem{Example}[Def]{Example}

\newcommand{\Thanks}{\vspace*{.5em} \noindent \thanks}
\newcommand{\beq}{\begin{equation}}
\newcommand{\eeq}{\end{equation}}
\newcommand{\Proof}{\begin{proof}}
\newcommand{\QED}{\end{proof} \noindent}
\newcommand{\QEDrem}{\ \hfill~$\Diamond$}

\newcommand{\la}{\langle}
\newcommand{\ra}{\rangle}

\newcommand{\Sl}{\mbox{$\prec \!\!$ \nolinebreak}}
\newcommand{\Sr}{\mbox{\nolinebreak$\succ$}}

\newcommand{\C}{\mathbb{C}}
\newcommand{\R}{\mathbb{R}}

\newcommand{\N}{\mathbb{N}}

\renewcommand{\H}{\mathscr{H}}
\newcommand{\U}{{\rm{U}}}

\newcommand{\SO}{{\rm{SO}}}

\newcommand{\bep}{\begin{pmatrix}}
\newcommand{\enp}{\end{pmatrix}}

\renewcommand{\O}{\mathscr{O}}

\renewcommand{\O}{{\mathscr{O}}}

\newcommand{\scrL}{\myscr L}
\newcommand{\itemD}{\item[{\raisebox{0.125em}{\tiny $\blacktriangleright$}}]}

\DeclareFontFamily{OT1}{rsfso}{}
\DeclareFontShape{OT1}{rsfso}{m}{n}{ <-7> rsfso5 <7-10> rsfso7 <10-> rsfso10}{}
\DeclareMathAlphabet{\myscr}{OT1}{rsfso}{m}{n}

\renewcommand{\Tr}{\text{\rm{Tr}}}
\DeclareMathOperator{\tr}{tr}

\newcommand{\bitem}{\begin{itemize}[leftmargin=2em]}
\newcommand{\eitem}{\end{itemize}}

\renewcommand{\sc}{\text{\rm{sc}}}
\DeclareMathOperator*{\slim}{s-lim}

\newcommand{\coloneqq}{:=}
\newcommand{\gammag}{\gamma_{g}}
\newcommand{\Reg}{\mathfrak{R}_{\varepsilon}}
\newcommand{\Cintegralwithn}{\frac{\xi^{k}_{p}(n)}{g(\xi_{p}(n),\nu)}\pi(\nu^{\perp})^{j}_{k}g(\xi_{p}(n),\partial_{j}u_{p})}
\newcommand{\Cintegral}{\frac{\xi^{k}_{p}}{g(\xi_{p},\nu)}\pi(\nu^{\perp})^{j}_{k}g(\xi_{p},\partial_{j}u_{p})}
\newcommand{\Cintegralwithoutp}{{\frac{\xi^{k}}{g(\xi,\nu)}\pi(\nu^{\perp})^{j}_{k}g(\xi,\partial_{j}u)}}
\newcommand{\Constantone}{\pi(\nu^{\perp})^{j}_{k}g_{mn}}
\newcommand{\volDpL}{\frac{1}{\mu_{p}(D_{p}\scrL)}}
\newcommand{\intDpL}{\int_{D_{p}\scrL}}
\newcommand{\Pkernel}{P^{\varepsilon}}

\begin{document}
\begin{abstract}
Based on a mechanism originally suggested for causal fermion systems, the present paper paves the way for a rigorous treatment of baryogenesis in the language of differential geometry and global analysis. Moreover, a formula for the rate of baryogenesis in Minkowski spacetime is derived.
\end{abstract}
\maketitle
\tableofcontents

\section{Introduction}
The Dirac operator is the central object of spin geometry. In the Lorentzian setting,
the Dirac equation describes the dynamics of spinor fields in spacetime.
Physically, spinor fields model matter on the quantum level.
A particular effect described by the Dirac equation is a process of pair creation, in which
both a particle and a corresponding anti-particle are created.
In this process, particles and anti-particles can be
created only in pairs (pair creation), but not individually.
Therefore, the Dirac equation cannot explain the asymmetry between matter and anti-matter
as observed in the Universe.
In order to explain this asymmetry, most physicists believe that
the Universe went through a process of matter creation during an early stage of its history. This process
is referred to as {\em{baryogenesis}}. Explaining if and why baryogenesis occurs is one of the main open questions in modern physics. 

Various mechanisms of {\em{baryogenesis}} have been proposed
(for a survey see for example~\cite{riotto} or~\cite{cline2006baryogenesis}).
Here we focus on the recent mechanism suggested in~\cite{baryogenesis},
which is based on a specific modification of the Dirac dynamics.
One goal of the present paper is to formulate this modified Dirac dynamics
and the resulting mechanism of baryogenesis in the language of differential geometry and global analysis,
thereby laying the foundations for a detailed mathematical treatment.
With this in mind, the first part of our paper (Sections~\ref{sec:preliminaries}--\ref{sec:dynamics_reg})
provides a self-contained introduction for mathematicians
working in Lorentzian and/or spin geometry.
Another goal of this paper is to quantify the effect of the mechanism in Minkowski spacetime. More specifically, 
in the second part of the paper (Section~\ref{sec:rate_of_baryo_Minkowski}) we show how to compute a rate of baryogenesis for Minkowski spacetime starting from the abstract mathematical framework developed in the first part. One of the motivations for this derivation is to set the stage for the study
of more general spacetimes (e.g.\ FLRW or inflationary spacetimes). This is the first of a series of papers where the baryogenesis mechanism will be worked out in detail for different cosmologically relevant spacetimes.

The modification of the Dirac dynamics proposed in~\cite{baryogenesis} stems from the
theory of causal fermion systems. Taking results from this theory as external input, we here
formulate baryogenesis in the familiar setting of Lorentzian spin geometry.
More precisely, we take the following input from the causal fermion systems theory:
\bitem
\itemD The Dirac equation should not hold on all length scales, but only down to a ``minimal length scale''
$\varepsilon$ (which can be thought of as the Planck length~$\approx 10^{-34}$ meters).
This is implemented mathematically by {\em{regularizing}} all spinor fields, in the simplest case by
mollifying them on the scale~$\varepsilon$.
\itemD The regularization is described by a globally defined timelike vector field~$u \in \Gamma(M, TM)$, the {\em{regularizing
vector field}}. This vector field satisfies dynamical equations which are related to the null geodesic flow
(the {\em{locally rigid dynamics of the regularizing vector field}}).
\itemD The dynamics of the spinor fields is obtained from the Dirac dynamics by
successive projections involving the regularizing vector field (so-called {\em{adiabatic projections}}).
The resulting dynamics, referred to as the {\em{locally rigid dynamics of the spinor fields}},
deviates slightly from the Dirac dynamics, giving rise to baryogenesis.
\eitem

More specifically, baryogenesis is described mathematically as follows.
Given a foliation~$(N_{t})_{t\in\mathbb{R}}$ and prescribing an initial timelike vector field~$u$ on a given hypersurface~$N_{t_{0}}$, the dynamics of this vector field is obtained from the null geodesic flow
forming integrals over the light cone (see Definition~\ref{def:locally_rigid_dynamics}
and the illustration in Figure~\ref{fig1}). The dynamics of the spinor fields is obtained
by modifying the Dirac dynamics by successive projections to the spectral subspace of
an operator~$A_t$ (cf.\ Definition~\ref{def:locally_rigid_dynamics_spinors}). The operator~$A_t$
is an essentially self-adjoint operator (see Lemma~\ref{lem:essen_sa}) which involves the regularizing vector field. It has the form
\begin{equation*}
A_{t} \coloneqq \frac{1}{4}\{u^{t},H_{g}+H_{g}^{\ast}\}+\frac{i}{4}\{u^{\mu},\nabla^{s}_{\mu}-(\nabla^{s}_{\mu})^{\ast} \} \:,
\end{equation*}
where~$H_g$ is the Hamiltonian and~$\nabla^{s}_{\mu}$ denotes the spin connection. The operator~$A_{t}$ is, arguably, the simplest symmetrized version of the Dirac Hamiltonian~$H_{g}$ which in addition allows the spinor fields to evolve according to a more general dynamics (i.e.\ not necessarily Dirac). The different possible spinor dynamics are described by the timelike vector field~$u$ (e.g.\ the Dirac dynamics corresponds to a vector field~$u$ which is normal at every point to the leaves of the chosen foliation~$(N_{t})_{t\in\mathbb{R}}$, see Remark~\ref{rem:specific_At}). We call~$A_{t}$ the \emph{symmetrized Hamiltonian}. Note that the family~$(A_{t})_{t\in\mathbb{R}}$ depends by construction on the choice of the time function~$t$.
In this setting, baryogenesis is described
by a relative change between suitably chosen subspaces of spinors. The choice of these subspaces is not obvious and will be introduced in Definition~\ref{def:regularized_spinor_space} (see also the motivating Remark~\ref{rem:spinor_spaces} and the motivating Example~\ref{rem:Baryo_Minkowski}).

Our main result is Theorem~\ref{theo:baryo_rate_Minkowski} which gives a simple formula for the rate of baryogenesis in Minkowski spacetime. The setting is the following. We begin with a foliation~$(N_{t})_{t\in\mathbb{R}}$ given by the level sets of the global time function~$t$ and a timelike vector~$u$ which initially is only defined on a Cauchy hypersurface~$N_{t_{0}}$ (with unit normal~$\nu$) by
\begin{equation*}
    u_{p}=(1+\lambda\Tilde{f}_{p})\nu+\lambda X_{p} \qquad \text{for all~$p\in N_{t_{0}}$}\:,
\end{equation*}
where~$\Tilde{f}$ is a smooth positive function taking values on~$N_{t_{0}}$, $X$ is a spacelike vector field and~$\lambda>0$. In practice, in order to guarantee that the coefficients of~$\frac{dA_{t}}{dt}$ are compactly supported in~$N_{t}$, we will afterward assume that~$u$ only has this specific form inside a compact subset~$V\subset N_{t_{0}}$, whereas on~$N_{t_{0}}\setminus V$ it simply agrees with the unit normal~$\nu$. Assuming the vector field~$u$ follows the locally rigid dynamics, it can be extended to the whole spacetime and an approximate (to first order in~$\lambda$) locally rigid dynamical equation for~$u$ is derived (cf.\ Lemma~\ref{lem:dyn_eq}):
\begin{equation*}
         \frac{du_{p}}{dt}=-\textrm{grad}_{\delta}(f_{p}^{-1})+\frac{\lambda}{f_{p}^{3}}\left(\frac{f_{p}}{3}\textrm{div}_{\delta}\left(X_{p}\right)+4X_{p}(f_{p})\right)\nu+\mathcal{O}(\lambda^{2})
\end{equation*}
where~$f\coloneqq 1+\lambda \Tilde{f}$ and ~$\delta$ denotes the Euclidean metric in~$\mathbb{R}^{3}$. In order to use this evolution equation to determine the rate of baryogenesis, several more technical results will be required: we will give an explicit expression for the functional calculus associated to~$H_{\eta}$ (Lemma \ref{lem:spectral_measure}), show that certain product operators (involving the spectral measure of the Dirac Hamiltonian~$H_{\eta}$ and the operator~$\frac{dA_{t}}{dt}$) are trace class (Lemma~\ref{lem:trace_class}) and we will prove a useful distributional equation (Lemma~\ref{lem:sokhotski}). With these results at our disposal and under the assumption that~$A_{t}$ has an absolutely continuous spectrum, we will derive an expression for the rate of baryogenesis due to the locally rigid evolution of the spinors to linear and quadratic order in~$\lambda$ (Theorem \ref{theo:baryo_rate_Minkowski}). More importantly, we will show that the rate of baryogenesis vanishes to first order in~$\lambda$ whereas in general it is nonzero to second order in~$\lambda$.

Note that our result does not contradict the fact that, in the current Universe, an increase in the asymmetry between matter and antimatter is not detected. The key point is that the value of the rate of baryogenesis~$B_{t}$ in Minkowski spacetime depends crucially on the initially prescribed regularizing vector field~$u|_{N_{t_{0}}}$ and that this vector field is to be understood as part of the physical data describing a specific physical situation. If the regularizing vector field is constant (or almost constant) on a given Cauchy hypersurface~$N_{t_{0}}$ (e.g.\ imposing the natural choice~$u=\partial_{t}$), then the rate of baryogenesis vanishes (or is very small), cf.\ 
the short discussion after Theorem~\ref{theo:baryo_rate_Minkowski}.


The paper is structured as follows. In Section~\ref{sec:preliminaries} we recall the basic spin-geometric concepts that will play an important role in our study. In Section~\ref{sec:Dirac_dynamics} we focus on the usual setting in spin geometry, namely spinor fields evolving according to the Dirac dynamics in a general spacetime. The importance of this section is to lay the foundations for a generalization of the main mathematical concepts and considered spinor spaces of the Dirac setting to a wider class of possible dynamics. This extension from the Dirac dynamics to a more general dynamics takes place in Sections~\ref{sec:general_dynamics} and~\ref{sec:math_baryo}. Moreover, in Section~\ref{sec:math_baryo} the main mathematical concepts in order to study and quantify baryogenesis are introduced. The particular dynamics that we will consider, the locally rigid dynamics, is presented in Section~\ref{sec:dynamics_reg} and is obtained through adiabatic projections. In Section~\ref{sec:rate_of_baryo_Minkowski}, we apply the mathematical framework we have constructed to Minkowski spacetime. In particular, after some tedious computations, a simple formula for the rate of baryogenesis to first and second order in $\lambda$ will be derived. Finally, our study leaves open some interesting questions which will be discussed in Section~\ref{sec:discussion}. Appendix~\ref{sec:computations} provides the detailed computations of some of the
 formulas needed in the proof of Theorem~\ref{theo:baryo_rate_Minkowski}.
 
\section{Preliminaries on globally hyperbolic spin geometry}\label{sec:preliminaries}
In this paper, all spacetimes~$(M,g)$ are considered to be four-dimensional, smooth, oriented, time oriented and globally hyperbolic. We denote by~$t$ a global time function and the associated global smooth foliation by~$(N_{t})_{t\in \mathbb{R}}$. We use the convention~$(+,-,-,-)$ for the signature of the Lorentzian metric~$g$. Furthermore, small Latin indices~$j,k, \ldots$ (except 'i', which is reserved for the imaginary unit) run from~$0$ to~$3$ and denote spacetime coordinate indices, whereas small Greek indices label the spatial coordinates and run from~$1$ to~$3$. Moreover, whenever a foliation~$(N_{t})_{t\in\mathbb{R}}$ is fixed in the spacetime~$(M,g)$ and a mathematical object is given in local coordinates, we always choose coordinates~$(x^{j})_{j=0,\ldots, 3}$ such that~$x^0 = t$ coincides with the time function. Finally, as is customary, the Einstein summation convention is used throughout the paper.

We denote {\em{Clifford multiplication}} in~$(M,g)$ by~$\gamma_{g} : TM \otimes SM \rightarrow SM$ and, given an orthonormal basis~$(e_{j})_{j=0,\ldots,3}$, we denote~$\gamma_{gj}\coloneqq\gammag(e_{j})$. 
 A globally hyperbolic four-dimensional Lorentzian manifold satisfies the topological spin condition
(for details see~\cite{baum, lawson+michelsohn}). Haven chosen a spin structure,
we denote the spinor bundle by~$SM$. The associated fiber~$S_pM\cong\mathbb{C}^{4}$ at a spacetime point~$p \in M$ is a
four-dimensional complex vector space
endowed with an indefinite inner product of signature~$(2,2)$.
Clifford multiplication is symmetric with respect to the inner product~$\Sl .|. \Sr_{S_pM}$\footnote{For a manifold with a metric signature~$(p,q)$ an analogous indefinite inner product can be constructed. See~\cite[Proposition 1.1.20 and Corollary 1.1.21]{treude} for the proof and the symmetry properties of the Clifford multiplication associated to this inner product.}. We refer to~$S_pM$ as the {\em{spinor space}} and~$\Sl .|. \Sr_{S_pM}$ as the {\em{spin inner product}}.
Sections in the spinor bundle are called {\em{spinor fields}}.
The metric induces on the spinor bundle a unique metric {\em{spin connection}}~$\nabla^{s}$
(also referred to as the Levi-Civita spin connection). In local coordinates, the spin connection can be written as (cf.~\cite[Section~4]{intro} or~\cite{u22}):
\begin{equation}\label{eq:spin_coeffs}
    \nabla^{s}_{j}=\partial_{j}-iE_{j}-ia_{j} \:,
\end{equation}
where~$E_{j},a_{j}$ are linear operators on the spinor space. In the physics literature, the~$E_{j}$ are referred to as the \emph{spin coefficients}.

Finally, the {\em{Dirac operator}}~$D_g : C^{\infty}(M,SM)\rightarrow C^{\infty}(M,SM)$ is given in local coordinates by
\[ D_{g}=i\gamma_{g}^j \nabla^{s}_j \:. \]
For smooth spinor fields (i.e.\ smooth sections of the spinor bundle) the {\em{Dirac equation}} of mass~$m\in [0,\infty)$ reads
\[  (D_g-m)\psi=0 \:. \]
Consider a foliation~$(N_{t})_{t\in \mathbb{R}}$ of the globally hyperbolic spacetime~$(M,g)$ and the space~$C^{\infty}_{\textrm{sc}}(M,SM)$ of smooth spinor fields with spatially compact support (i.e.\ for any~$N_{t'}\in (N_{t})_{t\in \mathbb{R}}$ and~$\psi \in C^{\infty}_{\textrm{sc}}(M,SM)$,
the support~$\textrm{supp}(\psi|_{N_{t'}})\subset SM$ is compact). Then, for any Cauchy hypersurface~$N_{t}$ with future directed normal~$\nu$ and for all~$\psi,\phi\in C^{\infty}_{\textrm{sc}}(N_{t},SM)$ that are solutions to the Dirac equation of mass~$m$, we define the scalar product~$(\psi|\phi)_{t}$ at a time~$t$ as follows:
\begin{equation}\label{eq:innerprod_hypersurface}
    (\psi|\phi)_{t}\coloneqq \int_{N_{t}} \Sl\psi|\gammag(\nu)\phi\Sr_{S_{p}M} d\mu_{N_{t}} \:,
\end{equation}
where~$d\mu_{N_{t}}$ is the Riemannian volume measure on~$N_{t}$ (induced by the Lorentzian volume measure~$d\mu_{M}$ in~$(M,g)$). It can be proven that, whenever the considered spinor fields solve the Dirac equation, this inner product is actually independent of the considered Cauchy hypersurface (see~\cite[equation (2.6)]{finite} or~\cite[Corollaries~2.1.3 and~2.1.4]{treude}) due to current conservation\footnote{By this we mean that~$(\psi|\gamma_{g}^{j}\varphi)$ is divergence-free if~$\psi,\varphi$ satisfy the Dirac equation. Then, integrating and applying the Gauss divergence theorem it directly follows that~$(\psi|\varphi)_{t_{0}}=(\psi|\varphi)_{t_{1}}$ for any~$t_{0},t_{1}\in\mathbb{R}$.}. In this paper, unless stated otherwise, whenever an adjoint operator is introduced it is understood that it is the adjoint with respect to the scalar product~$(.|.)_{t}$ given by~\eqref{eq:innerprod_hypersurface}.

\section{The Dirac dynamics and its regularization}\label{sec:Dirac_dynamics}

In this section we will focus on sections of the spinor bundle which follow the Dirac dynamics. We start by introducing some of the spinor spaces that will be relevant in this paper.

\begin{Def}\label{def:spinor_spaces}
    Let~$(M,g)$ be a globally hyperbolic spacetime with spinor bundle~$SM$ and consider a foliation~$(N_{t})_{t\in\mathbb{R}}$. Then we define the following relevant spinor spaces:
    \bitem
        \item[{\rm{(i)}}] Taking initial data in~$C^\infty_0(N_t, SM)$, we get smooth Dirac solutions with spatially compact support~$\psi \in C^\infty_\sc(M, SM)$. Taking the completion with respect to the scalar product~\eqref{eq:innerprod_hypersurface} gives a Hilbert space denoted by~$\H_m$.
Due to current conservation, the scalar product does not depend on time.
        \item[{\rm{(ii)}}] We denote by~$\H_{t}$ the space of square integrable spinor fields on the Cauchy hypersurface~$N_{t}$, i.e.\ 
\[ \H_{t}\coloneqq L^{2}(N_{t},SM) \]
    \eitem
\end{Def} \noindent
We note that spinor fields in~$\H_m$ are weak solutions of the Dirac equation.
More specifically, they are in the Sobolev space~$H^{1,2}_{\text{loc}}(M,SM)$.
Applying the trace theorem, their restriction to a Cauchy surface is well-defined in~$L^2_\text{loc}(N_t, SM)$.
Moreover, it is even in~$L^2(N_t, SM)$, and the scalar product~$(.|.)_t$ is well-defined.

\par Given a foliation~$(N_{t})_{t\in\mathbb{R}}$ and initial data~$\psi_{0}\in C^{\infty}_{0}(N_{t_{0}},SM)$, where~$t_{0}\in\mathbb{R}$, the following Cauchy problem for the Dirac equation
\[ (D_{g}-m)\psi=0 \quad \text{with} \quad \psi|_{N_{t_{0}}}=\psi_{0}\:, \]is known to be well-posed (cf.~\cite[Theorems~4.5.1 and 4.5.2]{intro}). By solving this Cauchy problem and evaluating the spinor field at a later time~$t$, we obtain a unitary evolution operator~$U_{t_{0}}^{t}$. Moreover, this operator is also well-defined and unitary for spinor fields in~$\H_{m}|_{N_{t_{0}}}$ and will play an important role in the study of the baryogenesis mechanism.

\begin{Def}\label{def:Hamiltonian}
    Consider a foliation~$(N_{t})_{t\in\mathbb{R}}$ of the globally hyperbolic spacetime~$(M,g)$ and fix a specific Cauchy hypersurface~$N_{t_{0}}$. The following operators are introduced:
    \bitem
        \item[{\rm{(i)}}] The \textbf{Dirac evolution operator}, 
\[ U_{t_{0}}^{t} : \H_{m}|_{N_{t_{0}}}\rightarrow \H_{m}|_{N_{t}} \]
        is a unitary operator that restricts 
        weak solutions to the Dirac equation to different Cauchy hypersurfaces. Moreover, a spinor field~$\psi\in H^{1,2}(M,SM)$ follows the Dirac dynamics if, there exists a suitable initial~$\psi_{0}\in \H_{t_{0}}$ such that~$\psi|_{N_{t}}=U_{t_{0}}^{t}\psi_{0}$.
        \item[{\rm{(ii)}}] The \textbf{Dirac Hamiltonian} 
        \begin{equation*}
            H_{g}: C^{\infty}_{0}(N_{t},SM)\subset \H_{t}\rightarrow \H_{t}
        \end{equation*} 
        is an operator  which, in local coordinates, is of the form:
        \begin{equation}\label{eq:Dirac_Hamiltonian}
            H_{g}\coloneqq -(\gammag^{0})^{-1}\left(i\gammag^{\mu}\nabla^{s}_{\mu}-m\right) - E_{0}-a_{0}
        \end{equation}
    \eitem
\end{Def}
In the next section we will consider spinor fields that evolve according to a more general dynamics than the Dirac dynamics.
\begin{Remark}\label{rem:Hamiltonian}\em{
Let us remark the following features of the Dirac Hamiltonian:
\bitem
\item[{\rm{(i)}}] The defining equation~\eqref{eq:Dirac_Hamiltonian} might seem familiar. Indeed, it is simply obtained by reordering terms in the Dirac equation~$(D_{g}-m)\psi=0$ using the coordinate expression~\eqref{eq:spin_coeffs} for the spin connection and defining the Dirac Hamiltonian as~$H_{g}\coloneqq i\partial_{t}$.
\item[{\rm{(ii)}}] In a general globally hyperbolic spacetime~$(M,g)$, the Hamiltonian~$H_{g}$ is not symmetric unless the spacetime is stationary (see~\cite[Section~4.6]{intro} or the detailed computation of the symmetry of~$H_{g}$ in~\cite[page 4]{chernoff}). However, for general globally hyperbolic spacetimes, the Dirac Hamiltonian is still a suitable starting point in order to derive a mathematical description of baryogenesis (more in Example~\ref{rem:Baryo_Minkowski}) because of the physical interpretation of the eigenvalues and eigenstates of~$H_{g}$ in stationary spacetimes. In stationary spacetimes, after constructing a suitable selfadjoint extension,
the Dirac equation can be solved by exponentiation, i.e.\ $\psi(t) = e^{-i t H_g}\, \psi_0$ with~$\psi_0 \in C^\infty_0(N_{t_0}, SM)$. Due to finite propagation speed, the resulting solution has spatially compact support,
$\psi \in C^\infty_\sc(M, SM)$. Identifying the Cauchy data with the corresponding solutions, one can consider the
Dirac Hamiltonian as an operator
\[ H_g : C^\infty_\sc(M, SM) \cap \H_m \rightarrow C^\infty_\sc(M, SM) \cap \H_m \:. \]

\vspace*{-1.7em}

\hfill\QEDrem
\eitem }
\end{Remark}

Often, we are only interested in a closed subspace of~$\H_m$, which we denote by~$\H$. This subspace models the particular physical situation (or spinor content) we
want to describe in the spacetime~$(M,g)$. Specific examples for the subspace~$\H$ are given in~\cite[Examples~2.3 and 2.6]{topology} and we will also choose such a Hilbert space in Remark~\ref{rem:Baryo_Minkowski} and Definition~\ref{def:Baryo}.

Note that, by construction, spinor fields in~$\H_{m}$ are in general only weak or distributional solutions to the Dirac equation. Thus, in order to evaluate the spinor fields in~$\mathcal{H}_{m}$ pointwise, they need to be regularized. For this purpose, the regularization operators are very useful. For the sake of completeness we recall their precise definition following~\cite[Definition~4.1]{finite} (see also~\cite[Definition~1.2.3]{cfs} in the context of Minkowski spacetime).

\begin{Def}[Regularization operator]\label{defreg} 
Let~$\H$ be a closed subspace of~$\H_{m}$. Then, a family of bounded linear operators~$({\mathfrak{R}}_\varepsilon)_{\varepsilon>0}$
which map~$L^{2}(M,SM)$ to the continuous spinor fields,
\[ {\mathfrak{R}}_\varepsilon \::\: \H \rightarrow C^0(M, SM)\cap \H_m \:, \]
are called {\bf{regularization operators}} if the following conditions hold:
\begin{itemize}[leftmargin=2em]
\item[{\rm{(i)}}] For every~$\varepsilon>0$ and~$p\in M$ there exists a constant~$c>0$
such that
\[ 
\| \big({\mathfrak{R}}_\varepsilon \psi \big)(p) \|_{S_{p}M} \leq c \:\|\psi\|_{t}\ \:, \]
where here~$\|.\|_{S_{p}M}$ denotes any norm on~$S_{p}M$.
\item[{\em{(ii)}}] In the limit~$\varepsilon\to 0^{+}$ the regularization operators go over to the identity with strong convergence of~$\mathfrak{R}_{\varepsilon}$ and~$\mathfrak{R}_{\varepsilon}^{\ast}$, i.e.\ for all~$\psi\in L^{2}(M,SM)$ it holds that
\[ 
\mathfrak{R}_{\varepsilon}\psi,\mathfrak{R}_{\varepsilon}^{\ast}\psi\to \psi \:\:{\text{as}} \:\: \varepsilon\to0^{+} \]
\end{itemize}
\end{Def} \noindent
Specific examples of regularization operators are mollifiers (\cite[Example 1.2.4]{cfs}, \cite[Section 4]{finite}), Fourier transforms (\cite[equation (5.5.2)]{intro}) or the 'hard cut-off' regularization which we will discuss in  Example~\ref{ex:hard_cut_off}. Note that these constructions are very similar to the smoothing operators or Friedrichs mollifiers used to map~$L^{2}$ sections (over general vector bundles) to smooth sections, see, e.g.~\cite[Definitions~1.4.8 and 1.4.10]{baer-spin} and~\cite[Definitions~5.19 and 5.20]{roe-elliptic}.

A regularization operator~$\Reg$ can also be considered as an operator defined on~$M$ such that for any~$p\in M$ we have~$\Reg(p) : \H \rightarrow S_{p}M$ with~$\Reg(p)\psi\coloneqq (\Reg\psi)(p)$. In other words, instead of analyzing the dynamics of the spinor fields in the image of~$\Reg$, we can alternatively focus on how~$\Reg$ changes from one point to another of the spacetime. Note that in some cases, given a point~$p\in M$ with coordinates~$(x^{j})_{j=0,\ldots,3}$, we will write~$\Reg(x)$ instead of~$\Reg(p)$ in order to make the coordinate dependence explicit. In Section~\ref{sec:reg_Dirac_dynamics} we study the Dirac dynamics of the regularization operators (and, thus, of the regularized spinor fields) in more detail.

\par In the next remark we discuss an important concept which is related to the regularization operators. Even though it will not play a central role in this paper, it will appear again in Section~\ref{sec:dynamics_reg} when discussing the Dirac dynamics of the regularized spinor fields or in the computations of Section~\ref{sec:rate_of_baryo_Minkowski}.

\begin{Remark}\label{rem:fermionic_proj}
\em{Consider a globally hyperbolic spacetime~$(M,g)$, a positive parameter~$
    \epsilon>0$ and a closed subspace~$\H\subset\H_{m}$. Then, the following operator maps smooth compactly supported spinor fields to regularized solutions to the Dirac equation:
\[ P^{\varepsilon}\coloneqq -\mathfrak{R}_{\varepsilon}\pi_{\H}(\mathfrak{R}_{\varepsilon})^{\ast}k_{m} \quad: \quad C^{\infty}_{0}(M,SM)\rightarrow\H_{m}\cap C^{0}(M,SM) \:, \]
    where~$\pi_{\H}$ is the orthogonal projection operator onto~$\H$, and the operator 
\[ k_{m} :	C_{0}^{\infty}(M,SM)\rightarrow C^{\infty}_{sc}(M,SM)\cap \H_{m} \]
    is called the \emph{causal fundamental solution} and corresponds to the difference of the advanced and the retarded Green operators (see e.g.~\cite[Definition~3.4.1 and Theorem~3.4.5]{baer+ginoux}, for more details) of the Dirac operator. The importance of the causal fundamental solution~$k_{m}$ is that it encodes the Dirac dynamics of the spinors and, moreover, allows to represent  solutions to the Dirac equation as follows. Any spinor field~$\psi\in C^{\infty}_{\textrm{sc}}(M,SM)$ which solves the Cauchy problem for the Dirac equation can be represented at a point~$p\in M$ with coordinates~$x=(x^{j})_{j=0,\ldots,3}$ by the formula
\[ 
        \psi(x)=2\pi\int_{N_{t_{0}}} k_{m}(x, y) \big(\gamma_{g}(\nu)\psi|_{N_{t_{0}}}\big)(y)d\mu_{N_{t_{0}}}\;, \]
where~$N_{t_{0}}$ is the Cauchy hypersurface on which the initial condition for the Cauchy problem is specified,~$\nu$ denotes its normal vector field and~$k_{m}(x,y)$ corresponds to the integral kernel of~$k_{m}$. For the proof of the previous representation formula see \cite[Proposition 13.6.1]{intro} or \cite[Lemma 2.1]{finite}. In Lemma~\ref{lem:spectral_measure} an explicit expression for the kernel~$k_{m}(x,y)$ in Minkowski spacetime (cf.\ equation~\eqref{eq:fund_solution_Dirac_Minkowski}) will be used in order to construct the projection valued measure of the Dirac Hamiltonian $H_{\eta}$.

    Coming back to the operator~$\Pkernel$, it also corresponds (cf.~\cite[equation (5.7)]{lqg}) to an integral operator whose integral kernel is given by:
    \begin{equation}\label{eq:kernel_ferm_proj}
 	P^{\varepsilon}(x,y)=-\mathfrak{R}_{\varepsilon}(x)(\mathfrak{R}_{\varepsilon}(y))^{\ast} : S_{q}M\rightarrow S_{p}M
    \end{equation}
    where~$x,y$ are respectively the coordinates of two points~$p,q\in M$. Moreover, this integral kernel is a bi-distributional solution to the Dirac equation (it is even the unique smooth bi-distributional solution, see~\cite[Prop 5.1]{lqg}), i.e.\
    \begin{equation}\label{eq:dirac_eq_kernel_ferm_proj}
        (D_{g}-m)P^{\varepsilon}(x,y)=0
    \end{equation}
    for all~$p,q\in M$. Hence, from expression~\eqref{eq:kernel_ferm_proj} it becomes clear that by solving the previous distributional equation we determine the Dirac dynamics of the regularized spinor fields in~$(M,g)$. This scenario has already been studied in depth and, in particular, it has been proven (see~\cite{reghadamard}) that the solution~$P^{\varepsilon}(x,y)$ of equation~\eqref{eq:dirac_eq_kernel_ferm_proj} admits a specific type of power expansion (known as the \emph{regularized Hadamard expansion}). This will be the starting point in Section~\ref{sec:reg_Dirac_dynamics} where the Dirac dynamics of the regularized spinor fields will be discussed in more detail.}\hfill\QEDrem
\end{Remark}

In the following example we construct a regularization operator in Minkowski spacetime.

\begin{Example}[Hard cut-off]\label{ex:hard_cut_off}
    \em{Consider Minkowski spacetime~$(\mathbb{R}^{1,3},\eta)$ in Cartesian coordinates. In particular, we construct an example of a regularization operator using an integral operator,
    \begin{equation}\label{eq:example_hard_cutoff}
        (\mathfrak{R}_{\varepsilon}\psi)(x)\coloneqq \int_{M} K^{\varepsilon}(x,y)\:\psi(y)\:d\mu_{M}(y)\;,
    \end{equation}
    where~$x=(x^{j})_{j=0,\ldots,3}$ corresponds to the coordinates of a point~$p\in \mathbb{R}^{1,3}$ and the integral kernel~$K^{\varepsilon}(x,y)$ is given by
\[ K^{\varepsilon}(x,y)=\int \frac{d^{4}k}{(2\pi)^{4}}e^{-ik(x-y)}\Theta(1+\varepsilon k^{0}) \:, \]
    where~$k(x-y)$ is short notation for the (Minkowski) inner product between~$k$ and~$x-y$. This specific regularization operator is often dubbed 'hard cut-off' regularization: because of the Heaviside function~$\Theta(1+\varepsilon k^{0})$ the interval of integration over~$k^{0}$ is only~$(-\frac{1}{\varepsilon},0)$. Note that, in terms of the Fourier transform~$\mathcal{F}$ and its inverse~$\mathcal{F}^{-1}$ (acting on tempered distributions), expression~\eqref{eq:example_hard_cutoff} can be rewritten as
\[ (\mathfrak{R}_{\varepsilon}\psi)(x)=\left(\mathcal{F}^{-1}\left(\Theta(1+\varepsilon k^{0})(\mathcal{F}\psi)\right)\right)(x) \:, \]
    so~$\Reg\psi\in C^{0}(M,SM)$ (by proceeding analogously at all~$p\in M$). In the particular case that~$\mathfrak{R}_{\varepsilon}$ is a hard cut-off regularization operator and~$\H \subset \H_m$ the subspace of negative-energy solutions to the Dirac equation, the integral kernel~$P^{\varepsilon}(x,y)$
in~\eqref{eq:kernel_ferm_proj} is computed to be
\beq \label{Pepsdef}
P^{\varepsilon}(x,y)=\int \frac{d^{4}k} {(2\pi)^{4}}\: \big( \gamma_\eta^{j}k_{j} + m \big)\: \delta(k^{2}-m^2)\:
\Theta(-k^{0})\: \Theta(1+\varepsilon k^{0}) \:e^{-ik(x-y)} 
\eeq
For the derivation of the previous expression see~\cite[Lemma 4.2]{lqg}, taking into account that instead of~$\Theta(1+\varepsilon k^{0})$ they consider a smooth regularizing factor~$e^{-\varepsilon k^{0}}$. Note that throughout this paper the regularization parameter~$\varepsilon$ and the mass~$m$ will be assumed to satisfy that~$m\ll\frac{1}{\varepsilon}$. Finally, we also introduce:
\begin{equation}\label{Pepsdef_omega}
P^{\varepsilon,\omega}(x,y)=\int \frac{d^{4}k} {(2\pi)^{4}}\: \big( \gamma_\eta^{j}k_{j} + m \big)\: \delta(k^{2}-m^2)\:
\Theta(-k^{0}+\omega)\: \Theta(1+\varepsilon k^{0}) \:e^{-ik(x-y)} 
\end{equation}
\hfill\QEDrem} 	
\end{Example}
\noindent The previous example will play an important role in Section~\ref{sec:rate_of_baryo_Minkowski} in order to construct the spectral measure associated to the Dirac Hamiltonian (cf.\ Lemma~\ref{lem:spectral_measure}) and when computing the rate of baryogenesis (see Theorem~\ref{theo:baryo_rate_Minkowski}).

\par The regularization operators lie at the core of the theory of causal fermion systems. In particular, only spinor fields in the image of~$\Reg$ are considered \emph{physical}: in order to study the change of certain mathematical properties of the spinor fields with respect to a fixed global time function~$t$, they need to be evaluated pointwise. So, instead of focusing on the spinor fields in~$\H_{m}$ or~$\H_{t}$, the main object of study is the {\em{regularized}} spinor fields. In this paper we will be mostly interested in specific spaces of continuous spinor fields, which will be introduced in the next section and will be denoted by~$\H^{\varepsilon}_{t}$ (cf.\ Definition~\ref{def:regularized_spinor_space}). In order to highlight more explicitly the link between the regularization operators and the spinor spaces introduced previously, in the following remark we construct a space of regularized spinor fields~$\H^{\varepsilon}_{t}$ using the regularization operators.

\begin{Remark}\label{rem:spinor_spaces}\em{
    Consider a general globally hyperbolic spacetime~$(M,g)$ with a foliation~$(N_{t})_{t\in\mathbb{R}}$, a regularization operator~$\mathfrak{R}_{\varepsilon}$ and a closed subspace~$\H\subset \H_{m}$. Then, keeping the spinor spaces of Definition~\ref{def:spinor_spaces} in mind, one possibility to define a space of regularized spinor fields (which we will denote by~$\H^{\varepsilon}_{t_{0}}$) is as follows
\[ \H^{\varepsilon}_{t_{0}}\coloneqq (\mathfrak{R}_{\varepsilon}\H)|_{N_{t_{0}}}\cap \H_{t_{0}} \:, \]
    where~$N_{t_{0}}\in(N_{t})_{t\in\mathbb{R}}$. If~$U_{t_{0}}^{t}$ corresponds to the unitary Dirac dynamics operator, then~$\H^{\varepsilon}_{t}\coloneqq U_{t_{0}}^{t}(\H^{\varepsilon}_{t_{0}})$ would correspond to the regularized spinor fields which evolved from a time~$t_{0}$ to a time~$t$ following the Dirac equation dynamics. In other words,
    \begin{equation*}
        (D_{g}-m)\psi=0
    \end{equation*}
    for all spinor fields~$\psi$ with~$\psi|_{N_t} \in \H^{\varepsilon}_{t}\coloneqq U_{t_{0}}^{t}(\H^{\varepsilon}_{t_{0}})$.\hfill\QEDrem}
\end{Remark}

\par Note that even though the Dirac dynamics of the regularized spinor fields is well understood, cf.~\cite{reghadamard}, this is the usual spin geometric setting in which no baryogenesis takes place. Furthermore, in most quantum geometry models the Dirac equation is only considered to describe the approximate evolution of the spinor fields. For this reason, we are interested in the case in which the regularized spinor fields evolve according to a slight deviation from the Dirac dynamics. In Section~\ref{sec:locally_rigid_dynamics}, we will turn our attention to the \emph{locally rigid dynamics}, in which small modifications of the Dirac dynamics are introduced through adiabatic projections. However, before discussing this specific type of dynamics, in the next section we set the stage and introduce the concepts required to study more general spinor dynamics.

\section{Going beyond the Dirac dynamics}\label{sec:general_dynamics}
As pointed out in~\cite[Section~6.1]{baryogenesis}, the Dirac dynamics does not give rise to baryogenesis.
Therefore, in this section we shall introduce some concepts that will allow us to study a more general spinor dynamics. However, note that the particular dynamics that we will study in more detail in this paper (the locally rigid dynamics), only introduces 'small' modifications of the Dirac dynamics (cf.\ Definition~\ref{def:locally_rigid_dynamics_spinors}). For the following definition it is useful to keep Remark~\ref{rem:spinor_spaces} in mind as a starting point to generalize the spinor spaces we introduced in the previous section.

\begin{Def}\label{def:regularized_spinor_space}
    Let~$(M,g)$ be a globally hyperbolic spacetime with spinor bundle~$SM$ and consider a foliation~$(N_{t})_{t\in\mathbb{R}}$. Then, given a hypersurface~$N_{t_{0}}\in(N_{t})_{t\in\mathbb{R}}$, a closed subspace~$\H_{t_{0}}^{\varepsilon} \subset C^{0}(N_{t_{0}},SM)\cap \H_{t_{0}}$ and an isometric
operator~$V_{t_{0}}^{t} : \H_{t_{0}}^{\varepsilon} \rightarrow C^{0}(N_{t},SM)\cap \H_{t}$, we define the {\bf{space of regularized spinor fields at a time~$\mathbf{t}$}} as
\[ \H^{\varepsilon}_{t}\coloneqq V_{t_{0}}^{t}(\H_{t_{0}}^{\varepsilon}) \]
\end{Def}
Note that as we demand the operator~$V_{t_{0}}^{t}$ to be isometric we only consider, by construction, dynamics for which the scalar product~$(.|.)_{t}$ is independent of the considered Cauchy hypersurface~$N_{t}$. This is
motivated by the commutator inner product for causal fermion systems (see~\cite[Section~3.4 and~4]{baryogenesis} for the details).

The meaning of the operator~$V_{t_{0}}^{t}$ is that it describes a general (i.e.\ not necessarily Dirac) evolution of the continuous 
spinor fields in the spacetime~$(M,g)$ analogously to how~$U_{t_{0}}^{t}$ describes their Dirac dynamics. In this paper the focus will be on a specific example of such an isometric operator~$V_{t_{0}}^{t}$, namely the locally rigid operator (see Definition~\ref{def:locally_rigid_dynamics_spinors}), which is constructed through adiabatic projections. The relation between the spinor spaces~$\H_{m}$,~$\H_{t}$ and~$\H^{\varepsilon}_{t}$ in the case of the Dirac dynamics was already discussed in Remark~\ref{rem:spinor_spaces}.

\section{Mathematical description of baryogenesis}\label{sec:math_baryo}

In the following remark we discuss how to describe mathematically a process of creation of particles (i.e.\ baryogenesis). As a guiding scenario, the simplest case, i.e.\ Minkowski spacetime, is analyzed with the aim of generalizing the discussion to arbitrary spacetimes. 

\begin{MotEx}\label{rem:Baryo_Minkowski} {\em{
    In Minkowski spacetime, the Dirac Hamiltonian
    \begin{equation*}
        H_{\eta}=-i\gamma_{\eta t}\gamma_{\eta}^{\alpha}\partial_{\alpha}+\gamma_{\eta t}m \: : \:C^{\infty}_{\textrm{sc}}(M,SM)\cap\H_{m}\rightarrow C^{\infty}_{\textrm{sc}}(M,SM)\cap \H_{m}
    \end{equation*} 
    is a self-adjoint operator with essential spectrum~$\sigma_{\text{ess}}=(-\infty,-m)\cup(m,\infty)$ (where~$m\in [0,\infty)$ is the mass parameter). \emph{Particles} (or \emph{anti-particles}) are eigenstates of~$H_{\eta}$ associated to positive (or negative) eigenvalues (which in the physics literature is referred to as the \emph{energy} of the particle). Moreover, when a spinor field that initially is an eigenstate corresponding to a negative eigenvalue~$E^{-}<-m$ of the Dirac Hamiltonian, at a later time, becomes an eigenstate corresponding to a positive eigenvalue~$E^{+}>m$, this is referred to as a process of \emph{particle creation}. Therefore, one option to describe mathematically baryogenesis in Minkowski spacetime is the following:
\bitem
    \item[{\rm{(i)}}] We choose~$\H$ as a subset of the negative spectral subspace of the Hamiltonian, i.e.\ $\H=(\chi_{(-\Lambda, -m)}(H_\eta))\H_{m}$, where~$\Lambda>m$ and~$\H_{m}$ is, as before, the completion of the space of smooth solutions to the Dirac equation with spatial compact support. Moreover, we choose a distinguished foliation~$(N_{t})_{t\in\mathbb{R}}$ of Minkowski spacetime.
    \item[{\rm{(ii)}}] Next, given a regularization operator~$\mathfrak{R}_{\varepsilon}$ and starting from 
    \begin{equation*}
    \H^{\varepsilon}_{t_{0}}\coloneqq (\mathfrak{R}_{\varepsilon}\H)|_{N_{t_{0}}}\cap \H_{t_{0}} \:,
    \end{equation*}
    with~$N_{t_{0}}\in (N_{t})_{t\in\mathbb{R}}$, the spinor fields are assumed to evolve in time according to given evolution equations
    (in the simplest case, the Dirac equation) described by a unitary operator~$V_{t_{0}}^{t}$. Evaluating the regularized spinor fields on a later Cauchy surface, and recalling that~$\H^{\varepsilon}_{t}\coloneqq V_{t_{0}}^{t}(\H^{\varepsilon}_{t_{0}})$, it may happen that 
\beq \label{shift}
        \tr_{\H^{\varepsilon}_{t_{0}}} \big( \chi_{(-\Lambda, -m)}(H_\eta)) \big) \neq \tr_{\H^{\varepsilon}_{t}} \big( \chi_{(-\Lambda, -m)}(H_\eta) \big)
\eeq
     for~$t>t_{0}$. If this is the case, one can say that baryogenesis takes place.

We note for clarity that the spectral subspace~$\chi_{(-\Lambda, -m)}(H_\eta)$ is the same on both sides of the equation.
The subspace~$\H^{\varepsilon}_{t_{0}}$ has evolved to the new subspace~$\H^{\varepsilon}_{t}$.
The change of this subspace relative to the fixed subspace~$\chi_{(-\Lambda, -m)}(H_\eta)$ is quantified
by the trace~\eqref{shift}. Moreover, it will be important to verify that the difference in~\eqref{shift}
is independent of the value of~$\Lambda$. In this way, we make sure that we focus on the behavior
near the spectral point~$-m$.
\hfill\QEDrem\eitem}} \end{MotEx}

We would like to generalize the previous description to arbitrary globally hyperbolic spacetimes. In particular, it would be interesting to describe baryogenesis as a relative change between the spectral subspaces of a, yet to be defined, self adjoint operator and the space of regularized spinor fields~$\H^{\varepsilon}_{t}$. 
Nevertheless, two main limitations of the Dirac Hamiltonian~$H_{g}$ have to be circumvented:
\bitem
    \item[{\rm{(i)}}] As already discussed, the Dirac Hamiltonian~$H_{g}$ is not symmetric in a general spacetime. Hence,~$H_{g}$ is not suitable for our purposes as it doesn't allow to generalize the spectral theory arguments of the previous Remark.
    \item[{\rm{(ii)}}] We want to consider a more general dynamics than the Dirac dynamics. 
\eitem

This motivates the following definition which can be seen as, arguably, the simplest symmetrized generalization of the Dirac Hamiltonian~$H_{g}$. Furthermore, in order to incorporate a more general spinor dynamics, it introduces a globally defined timelike vector field~$u$. The idea is that it is this vector field which should encode the deviations from the Dirac dynamics.

\begin{Def} {\rm{(Symmetrized Hamiltonian)}}\label{def:symmetrized_Hamilt}
    Let~$(M,g)$ be a globally hyperbolic spacetime and~$(N_{t})_{t\in\mathbb{R}}$ a distinguished foliation where~$\nu|_{N_{t}}$ is the unit normal vector field to each hypersurface~$N_{t}$. Moreover, choose coordinates~$(x^{j})_{j=1,\ldots, 3}$ with~$x^0=t$ and consider a smooth global future directed timelike vector field~$u: M\rightarrow TM$, which will be referred to as the \textbf{regularizing vector field}. Then, we define the {\bf{symmetrized Hamiltonian}} at a time~$t$ as the operator 
\[ A_{t}: C^\infty_0(N_{t},SM) \subset \H_t \rightarrow \H_t \]
which in local coordinates is given by the following expression:
    \begin{equation}\label{eq:baryogenesis_operator}
    A_{t}\coloneqq \frac{1}{4}\{u^{0},H_{g}+H_{g}^{\ast}\}+\frac{i}{4}\{u^{\mu},\nabla^{s}_{\mu}-(\nabla^{s}_{\mu})^{\ast} \}
\end{equation} 
where~$H_{g}$ is the Dirac Hamiltonian.
\end{Def}
\noindent By construction, $(A_{t})_{t\in\mathbb{R}}$ is a family of symmetric operators acting
on~$C^\infty_0(N_{t},SM) \subset \H_t$. We now show that these operators are even essentially self-adjoint.

\begin{Lemma}\label{lem:essen_sa} For any~$t_0 \in \R$, the operator~$A_{t_0}$ with domain~$C^\infty_0(N_{t_0},SM) \subset \H_{t_0}$ is essentially self-adjoint on the Hilbert space~$\H_{t_0}$.
\end{Lemma}
\Proof For a real parameter~$\tau$ (which can be thought of as the time of an artificial static spacetime),
we consider the Cauchy problem on~$\R \times N_{t_0}$
\[ i \partial_\tau \psi(\tau) = A_{t_0} \psi(\tau) \:,\qquad \psi(0) = \psi_0 \in C^\infty_0(N_{t_0}, SM) \]
(where~$\psi(\tau) := \psi|_{\{\tau\} \times N_{t_0}}$).
Because~$A_{t_{0}}$ is a symmetric operator on~$C^\infty_0(N_{t_{0}},SM) \subset \H_{t_{0}}$, this Cauchy problem can be solved using the theory of linear symmetric hyperbolic systems
(see for example~\cite{john, taylor3} or~\cite[Chapter~13]{intro}).
Due to finite propagation speed, the solution is compactly supported for any~$\tau \in \R$,
i.e.\ $\psi(\tau) \in C^\infty_0 \big( \{\tau\} \times N_{t_0}, SM \big)$.
We thus obtain a family of ``time evolution operators''
\[ \U_{\tau', \tau} \::\: C^\infty_0 \big( \{\tau\} \times N_{t_0}, SM \big) \rightarrow C^\infty_0
\big( \{\tau'\} \times N_{t_0}, SM \big) \:, \]
which form a one-parameter group. Moreover, the computation
\[ \frac{d}{d\tau} ( \psi(\tau) \,|\, \phi(\tau) )_{t_0} = 
\big( (-i A_{t_0}) \psi(\tau) \,\big|\, \phi(\tau) \big)_{t_0} + \big( \psi(\tau) \,\big|\, (-i A_{t_0}) \phi(\tau) \big)_{t_0} = 0 \]
(where in the last step we used that~$A_{t_0}$ is a symmetric operator)
shows this group is formed of unitary operators.
This makes it possible to apply a result by Chernoff~\cite[Lemma 2.1]{chernoff73} to conclude
that the Hamiltonian is essentially self-adjoint with domain~$C^\infty_0(N_{t_0}, SM)$.
\QED
In what follows, we denote the unique self-adjoint extension of~$A_t$ by the same symbol.

\par In the following remark we show in what sense the symmetrized Hamiltonian~$A_{t}$ encodes general spinor dynamics and that, in some particular cases, it has a very simple form.

\begin{Remark}\label{rem:specific_At}\em{
Consider a distinguished foliation~$(N_{t})_{t\in\mathbb{R}}$ of a globally hyperbolic spacetime~$(M,g)$ and the unit length vector field~$u : M\rightarrow TM$ normal to any Cauchy hypersurface, i.e.\ $u|_{N_{t}}=\nu|_{N_{t}}=\partial_{0}|_{N_{t}}$ for all~$t\in \mathbb{R}$. Then, the symmetrized Hamiltonian simplifies to 
    \begin{equation*}
        A_{t}=\frac{1}{2}(H_{g}+H_{g}^{\ast})
    \end{equation*}
i.e.\ $A_{t}$ is nothing more than the symmetrization of the Dirac Hamiltonian. Furthermore, if the spacetime is stationary, then~$H_{g}$ is a self adjoint operator and the previous expression simply becomes:
    \begin{equation*}
        A_{t}=H_{g}
    \end{equation*}
This makes it clearer in what sense the symmetrized Hamiltonian does incorporate the Dirac dynamics: given a foliation of~$(M,g)$, it suffices to choose a vector field~$u : M \rightarrow TM$ normal to each Cauchy hypersurface.
    
However, for a general timelike vector field~$u$ in, for example Minkowski spacetime~$(\mathbb{R}^{1,3},\eta)$ in Cartesian coordinates, the operator~$A_{t}$ is given by
\[ 
    A_{t}= \frac{1}{2}\{u^{0},H_{\eta}\}+\frac{i}{2}\{u^{\mu},\partial_{\mu}\}=\frac{1}{2}\{u^{0},-i\gamma_{\eta 0}\gamma_{\eta}^{\mu}\partial_{\mu}+\gamma_{\eta t}m\}+\frac{i}{2}\{u^{\mu},\partial_{\mu}\} \]
    
\vspace*{-1.5em}

\hfill\QEDrem}
\end{Remark}

As previously discussed, it seems convenient to describe mathematically baryogenesis as a shift in a subset of the spectrum of the family of self-adjoint operators~$(A_{t})_{t\in\mathbb{R}}$. This motivates the following definition for baryogenesis.
\begin{Def}\label{def:Baryo}{\rm{(Baryogenesis)}}.
    Let~$(M,g)$ be a globally hyperbolic spacetime with a distinguished foliation~$(N_{t})_{t\in\mathbb{R}}$. Furthermore, let~$\varepsilon,\Lambda>0$  and consider the family of symmetrized Hamiltonians~$(A_{t})_{t\in\mathbb{R}}$ with spectral measure~$\chi_{I}(A_{t})$ (where~$I\coloneqq (-1/\varepsilon, -m)$). In the first place, fix an initial subspace~$\H^{\varepsilon}_{t_{0}}\subset C^{0}(N_{t_{0}},SM)\cap \H_{t_{0}}$ with~$N_{t_{0}}\in(N_{t})_{t\in\mathbb{R}}$. Then, we say that {\bf{baryogenesis}} takes place at a time~$t_{1}>t_{0}$ provided that
    \begin{align}\label{eq:baryo_def}
        \text{\rm{tr}}_{\H^{\varepsilon}_{t_{1}}} \Big( \eta_{\Lambda}\Big(\frac{H_{g}+H_{g}^{\ast}}{2}\Big)\chi_{I}(A_{t_{1}}) \Big)
        \neq \textrm{\rm{tr}}_{\H^{\varepsilon}_{t_{0}}} \Big( \eta_{\Lambda}\Big(\frac{H_{g}+H_{g}^{\ast}}{2}\Big)\chi_{I}(A_{t_{0}}) \Big) \:,
    \end{align}
    where~$\eta_{\Lambda}\in C^{\infty}_{0}(\mathbb{R},[0,\infty))$ is a smooth cut-off operator compactly supported in~$(-\Lambda,\Lambda)$. Furthermore, we define the {\bf{rate of baryogenesis}} as
    \begin{equation}
        B_{t}\coloneqq \frac{d}{dt}\text{\rm{tr}}_{\H^{\varepsilon}_{t}} \Big( \eta_{\Lambda}\Big(\frac{H_{g}+H_{g}^{\ast}}{2}\Big)\chi_{I}(A_{t}) \Big)\;.
    \end{equation}
\end{Def}
The formula~\eqref{eq:baryo_def} can be understood as a generalization of~\eqref{shift} to arbitrary
globally hyperbolic spacetimes. One obvious difference is that the Dirac Hamiltonian~$H_\eta$ has been
replaced by the time-dependent operator~$A_t$. Therefore, the difference of the traces quantifies the relative change between the subspace~$\H^{\varepsilon}_t$ and the spectral subspace of~$A_t$,
which now are both time dependent.
Moreover, a technical difference is that we now prefer to work with a smooth cutoff near~$-\Lambda$. Note that throughout the paper we will assume that the parameters~$\varepsilon,m,\Lambda>0$ satisfy that
\begin{equation*}
    m\ll \Lambda\ll\frac{1}{\varepsilon}\:.
\end{equation*}
    

\section{The dynamics of the regularization} \label{sec:dynamics_reg}

\subsection{The Dirac dynamics}\label{sec:reg_Dirac_dynamics}
Our starting point is the discussion initiated in Remark~\ref{rem:fermionic_proj} on the Dirac dynamics of the regularized spinor fields in an arbitrary globally hyperbolic spacetime~$(M,g)$ following~\cite{reghadamard}. Consider a closed subspace~$\H\subset\H_{m}$ and regularization operators which map to solutions to the Dirac equation, i.e.\
\[ \Reg: \H \rightarrow C^{0}(M,SM)\cap \H_m \:. \]
Recall that, in this setting, it holds that for all~$p, q\in M$ (with coordinates~$x=(x^{j})_{j=0,\ldots,3}$ and~$y=(y^{j})_{j=0,\ldots,3}$ respectively), $P^{\varepsilon}(x,y)=-\mathfrak{R}_{\varepsilon}(x)(\mathfrak{R}_{\varepsilon}(y))^{\ast} : S_{q}M\rightarrow S_{p}M$ is a regularized bi-solution to the Dirac equation (i.e.\ $(D_g-m)P^{\varepsilon}(x,y)=0$). Thus it encodes the Dirac dynamics of the regularized spinor fields and it additionally satisfies that~$\lim_{\epsilon\to0}P^{\varepsilon}(x,y)$ is a distributional bi-solution to the Dirac equation. 

\par Furthermore, in~\cite{reghadamard} it was shown that this regularized bi-solution to the Dirac equation admits, locally, a regularized Hadamard expansion in each geodesically convex neighborhood. The physical importance of this Hadamard form is that the corresponding quantum state satisfies a micro-local energy condition~\cite{radzikowski}. For the purposes of this paper, we do not need to enter the details of Hadamard states and/or Hadamard expansion. It suffices to note that the regularized Hadamard expansion 
gives explicit information on the behavior of~$P^\varepsilon(x,y)$ near the lightcone
(i.e.\ for points~$x$ and~$y$ which are lightlike separated or whose geodesic distance is very small).
This expansion is formulated in terms of two functions~$\Gamma$ and~$f$ (cf.~\cite[equations (1.10)-(1.13) and (5.1)]{reghadamard}), as we now briefly summarize.

Let~$V$ be a geodesically convex neighborhood and~$\gamma : [\tau_{q},\tau_{p}]\subset\mathbb{R}\rightarrow M$ the unique geodesic (up to reparametrizations) from~$q=\gamma(\tau_{q})\in V$ to~$p=\gamma(\tau_{p})$, with length~$L_{g}(\gamma)$. Then~$\Gamma : V \times V \rightarrow \mathbb{R}$ is defined as the geodesic distance squared, i.e.\
\begin{equation*}
  \Gamma(p,q)\coloneqq\pm \big( L_{g}(\gamma) \big)^{2}=g \big(
  \textrm{exp}_{T_{q}M}^{-1}(p),\textrm{exp}_{T_{q}M}^{-1}(p) \big) \:,
\end{equation*}
where~$\textrm{exp} : TM \rightarrow M$ denotes the exponential map and with the convention of choosing the~$-$ (or~$+$) sign for spacelike (or timelike) separated points. By construction, $\Gamma$ vanishes for lightlike separated points (see also~\cite[Lemma~1.3.19]{baer+ginoux} or~\cite[Lemma B.1]{dgc} for some of the properties of~$\Gamma$). 

\par On the other hand,~$f : V \times V \rightarrow \mathbb{R}$ is a real valued function which is related to~$\Gamma$ by the following transport equations (cf.~\cite[Proposition~2.2]{reghadamard}; similar transport equations are also derived in~\cite[Section~2.1]{baer+ginoux}):
\begin{equation}\label{eq:transport_eq}                     2g(\textrm{grad}_{p}\Gamma(p,q),\textrm{grad}_{p}f(p,q))=4f(p,q)\hspace{0.8cm}\textrm{and}\hspace{0.8cm} f(q,q)=0
\end{equation}
for all~$p,q\in V$, where~$V$ denotes again a geodesically convex neighborhood in~$M$
(for our purpose it suffices to consider the case that~$p$ and~$q$ are lightlike separated or have small geodesic distance).
By solving these transport equations, the coefficients of the Hadamard expansion of~$P^{\varepsilon}(x,y)$ are determined and, in the light of the previous discussion, the Dirac dynamics of the regularized spinor fields in~$(M,g)$ is obtained.

\begin{Lemma}
Let~$V$ be a geodesically convex neighborhood and~$f$ a real valued function satisfying equation~\eqref{eq:transport_eq}. Given a geodesic~$\gamma : [\tau_{q},\tau_{p}]\subset \mathbb{R}\rightarrow M$ with~$\gamma(\tau_{q})=q\in V$ and~$\gamma(\tau_{p})=p\in V$, it holds that
\[ f(p,q)=\tau_{p}-\tau_{q} \:. \]
\end{Lemma}
\begin{proof}
Let~$\gamma : [\tau_{q},\tau_{p}]\subset \mathbb{R}\rightarrow M$ be a geodesic with~$\gamma(\tau_{q})=q$ and~$\gamma(\tau_{p})=p$. For an arbitrary~$\tau\in [\tau_{q},\tau_{p}]$ we denote:
\begin{equation}
    \Gamma_{q}(\gamma(\tau))\coloneqq\Gamma(q,\gamma(\tau)) \qquad \textrm{and}\qquad f_{q}(\gamma(\tau))\coloneqq f(q,\gamma(\tau))\nonumber
\end{equation}
We use the following expression from~\cite[equation (B.6)]{dgc}: 
\begin{equation}
    \textrm{grad}_{\gamma(\tau)}\Gamma_{q}(\gamma(\tau))=2 \tau\dot{\gamma}(\tau)\nonumber.
\end{equation}
Note that this result directly follows from~\cite[Lemma B.2]{dgc} or~\cite[Lemma 1.3.19 and equation (1.18)]{baer+ginoux}. We can then rewrite the transport equation~\eqref{eq:transport_eq} as follows
\begin{align}
     g\big(\tau\dot{\gamma}(\tau),\textrm{grad}_{\gamma(\tau)}(f_{q}(\gamma(\tau)))\big)&=f_{q}(\gamma(\tau))\nonumber\\
    \iff \tau \dot{\gamma}^{j}(\tau)\partial_{j}\big(f_{q}(\gamma(\tau))\big)&=f_{q}(\gamma(\tau))\nonumber
\end{align}
By the chain rule~$\frac{d}{d\tau}f_{q}(\gamma(\tau))=\dot{\gamma}^{j}(\tau)\partial_{j}\big(f_{q}(\gamma(\tau))\big)$, the previous expression can be rewritten as an ordinary differential equation which can be easily solved by separation of variables with the initial condition that~$f_{q}(\gamma(\tau_{q}))=f(q,q)=0$:
\begin{align}
    &\tau \frac{d}{d\tau}f_{q}(\gamma(\tau))=f_{q}(\gamma(\tau))\nonumber\\
     \iff\int_{0}^{f_{q}(p)}& \Big(\frac{d}{d\tau}f_{q}(\gamma(\tau))\Big) \frac{d\tau}{f(\gamma(\tau))}=\int_{\tau_{q}}^{\tau_{p}} \frac{d\tau}{\tau}\nonumber
\end{align}
which yields~$f_{q}(p)=\tau_{p}-\tau_{q}$.
\end{proof}
\par The importance of the previous lemma is that it implies that, if a point~$q\in M$ and a geodesically convex neighborhood~$V$ are fixed, at any other lightlike separated point~$p$ in~$V$ the value of~$f(p,q)$  is given by the difference in the value of the affine parameter of the unique (unparametrized) null geodesic~$\gamma$ joining them. However, because of the freedom to reparametrize the chosen null geodesic~$\gamma$, the function~$f(p,q)$ does not have a unique value. Note that, given a particular parametrized null geodesic~$\gamma : [\tau_{q},\tau_{p}]\rightarrow M$ with~$\gamma(\tau_{q})=q$ and~$\gamma(\tau_{p})=p$, additive reparametrizations of this geodesic do not change the value of~$f(p,q)$ (as its value is given by a \emph{difference} in the value of the chosen affine parameter). But multiplicative reparametrizations of~$\gamma$ do yield a different value of~$f(p,q)$. Hence, the function~$f$ solving the transport equations~\eqref{eq:transport_eq} is unique up to a multiplicative constant.

In order to circumvent the non-uniqueness of~$f$, a distinguished multiplicative parametrization for the null geodesics in~$(M,g)$ can be (and will be) fixed. Geometrically, this is the scenario we are interested in, as then the flow of these unique null geodesics in~$(M,g)$ encodes the unique Dirac dynamics of the regularized spinor fields in this spacetime.

In the following definition, a distinguished reparametrization for the null geodesics in~$(M,g)$ is chosen
as follows. Starting from a foliation~$(N_{t})_{t\in\mathbb{R}}$ and fixing a Cauchy hypersurface~$N_{t_{0}}$, a future-directed timelike vector field~$u$ on~$N_{t_{0}}$ is assumed. This timelike vector field can be used to
choose a distinguished set of null geodesics~$\gamma : I\subset \mathbb{R}\rightarrow M$ (and to fix a unique reparametrization) by demanding that~$g_{\gamma(s)}(u_{\gamma(s)},\dot{\gamma}(s))=1$ whenever~$\gamma(s)\in N_{t_{0}}$.

\begin{Def}\label{def:DxL}
    Let~$(N_{t})_{t\in\mathbb{R}}$ be a foliation of the  globally hyperbolic spacetime~$(M,g)$ and choose a Cauchy hypersurface~$N_{t_{0}}$. Furthermore, let~$u : N_{t_{0}}\rightarrow TM~$ be a smooth future directed timelike vector field. Then,~$\scrL$ is the set of maximally extended future directed null geodesics~$\gamma: I\subset \mathbb{R}\rightarrow M$ (together with the interval of parametrization~$I$) in~$(M,g)$ such that whenever~$\gamma(s)\in N_{t_{0}}$ it holds that
    \begin{equation}\label{eq:DxL}
        g_{\gamma(s)}(u_{\gamma(s)},\dot{\gamma}(s))=1
    \end{equation}
    Furthermore, for an arbitrary point~$p\in M$ we define the hypersurface~$D_{p}\scrL$ of the null bundle as
\[ D_{p}\scrL\coloneqq \{\dot{\gamma}(s)\; | \;(I,\gamma)\in \scrL \;\textrm{and}\; \gamma(s)=p\} \:. \]
\end{Def}
By definition, for a point~$q\in M\setminus N_{t_{0}}$ it holds that:
\begin{equation}\label{eq:DqL_parallel_transport}
    D_{q}\scrL = \{P^{\gamma}_{s_{0},s_{1}}\dot{\gamma}(s_{0})\::\:\dot{\gamma}(s_{0})\in D_{N_{t_{0}}}\scrL, \:\gamma(s_{1})=q \}
\end{equation}
where,~$D_{N_{t_{0}}}\scrL\coloneqq \bigcup_{p\in N_{t_{0}}} D_{p}\scrL$ and~$P^{\gamma}_{s_{0},s_{1}} : T_{\gamma(s_{0})}M\rightarrow T_{\gamma(s_{1})}M$ denotes the parallel transport map in~$(M,g)$ along the curve~$\gamma$. Writing~$D_{q}\scrL$ in terms of the parallel transport map as in expression~\eqref{eq:DqL_parallel_transport} has the advantage of highlighting that, upon the initial choice of~$u$ on~$N_{t_{0}}$, the specific form of these null hypersurfaces is fully determined by the geometry of~$(M,g)$. Note that~$D_{p}\scrL$ is a topological sphere (cf.~\cite[Appendix~A]{dgc}).

\begin{Remark}\label{rem:shape_DxL} {\em{
    Consider a Cauchy hypersurface~$N_{t_{1}}$ with~$t_{1}>t_{0}$ and a point~$q\in N_{t_{1}}$. Note that although, by definition, for each~$p\in N_{t_{0}}$ and each~$\dot{\gamma}(s)\in D_{p}\scrL$ it holds that~$g_{p}(u_{p},\dot{\gamma}(s))=1$, in general there does not exist anymore a timelike vector field~$u$ such that~$g_{q}(u_{q},\dot{\gamma}(s'))=1$ for all~$\dot{\gamma}(s')\in D_{q}\scrL$. Intuitively, this means that even if for~$p\in N_{t_{0}}$ the cross section~$D_{p}\scrL$ will have a clear ellipsoidal shape, the shape of the hypersurfaces~$D_{q}\scrL$ (with~$q\in N_{t_{1}}$) will deviate from the original ellipsoid and become more irregular (see the lower part of Figure~\ref{fig1}).
\input{fig1.tex}%
In particular, while~$D_{p}\scrL$ is fully determined by the value of~$u_{p}$,~$D_{q}\scrL$ (with~$q\in N_{t_{1}}$) will be determined by the initial~$u|_{N_{t_{0}}}$ and the parallel transport map (along the null geodesics in~$\scrL$) in~$M$ (cf. expression~\eqref{eq:DqL_parallel_transport}). Hence, in general there is no one-to-one correspondence anymore between a subset~$D_{q}\scrL$ (with~$q\in N_{t_{1}}$ and~$t_{1}>t$ arbitrary) of the null bundle and a timelike vector field at the same point.
}} \hfill\QEDrem
\end{Remark}

\subsection{The locally rigid dynamics} \label{sec:locally_rigid_dynamics}
The locally rigid dynamics aims to obtain small modifications to the Dirac dynamics through adiabatic projections. The starting point is the cross section of the light cone~$D_{p}\scrL$ with~$p\in N_{t_{0}}$ and the discussion in Remark~\ref{rem:shape_DxL}. The idea is to consider a small time-step~$\Delta t$ and to slightly perturb~$D_{q}\scrL$ such that this hypersurface is again completely determined by a timelike vector field at~$y$ satisfying condition~\eqref{eq:DxL}. 

\par We start by presenting the locally rigid dynamics of the regularizing vector field~$u$ and, afterward, discuss how the locally rigid dynamics of the spinor fields is obtained.

\begin{Def}[Locally rigid dynamics of~$u$]\label{def:locally_rigid_dynamics}
     Let~$(N_{t})_{t\in\mathbb{R}}$ be a foliation of the  globally hyperbolic spacetime~$(M,g)$ and choose a Cauchy hypersurface~$N_{t_{0}}$. Furthermore, let~$u : N_{t_{0}}\rightarrow TM~$ be a smooth future directed timelike vector field, and let~$\scrL$ and~$D_{q}\scrL$ (for an arbitrary point~$q\in M$) be as in Definition~\ref{def:DxL}. Consider a sufficiently small~$\Delta t$ such that for any~$q \in N_{t_{0}+\Delta t}$ there exists a normal neighborhood~$U\subset M$ of~$q$ with~$U \cap N_{t_{0}}\neq \emptyset$. Then we define the following timelike vector field at~$q$:
\[ \overline{\xi}_{q}\coloneqq\frac{1}{\mu_{q}(D_{q}\scrL)}\int_{D_{q}\scrL}\dot{\gamma}(s)\:d\mu_{q}(\dot{\gamma}(s)) \:, \]
     where~$d\mu_{q}(\dot{\gamma}(s))$ is the induced volume measure on~$D_{q}\scrL$,$(I,\gamma)\in\scrL$ and~$\gamma(s)=q$. Using the vector field~$\overline{\xi}_{q}$, we define the regularizing vector field at~$q$ by
\[ u_{q}\coloneqq\frac{1}{\overline{\xi}_{q}^{2}}\overline{\xi}_{q} \:. \]
    Proceeding in an analogous way at each~$q\in N_{t_{0}+\Delta t}$ the timelike vector field~$u$ is extended to~$N_{t_{0}+\Delta t}$. We refer to this process as the {\bf{locally rigid dynamics}} of~$u$.
\end{Def}
Note that, because by assumption~$\Delta t$ can be arbitrarily small, it might seem more convenient to refer in the previous definition to a \emph{short-time} locally rigid dynamics of~$u$. However, by the following argument, the previous definition can be extended to construct a \emph{global} locally rigid dynamics: let us proceed iteratively in~$k$ time steps~$\Delta t$ with (i.e.\ $k=0,1,,\ldots,k_{\max}-1$):
\begin{enumerate}
    \item[{\rm{(i)}}] First, at each~$q\in N_{t_{0}+(k+1)\Delta t}$ we define
            \begin{equation*}
                 \overline{\xi}^{k+1}_{q}\coloneqq\frac{1}{\mu_{q}(D_{q}\scrL^{k})}\int_{D_{q}\scrL^{k}}\dot{\gamma}(s)d\mu_{q}(\dot{\gamma}(s)) \qquad \text{and} \qquad u_{q}\coloneqq\frac{1}{\big|\overline{\xi}^{k+1} \big|^{2}}\:
\overline{\xi}^{k+1}_{q}\:.
            \end{equation*}
    \item[{\rm{(ii)}}] In the second step, the set~$\scrL^{k+1}$ is defined as the collection of maximal null geodesics that satisfy the condition
    \begin{equation*}
        g_{\gamma(s)}\big(u_{\gamma(s)},\dot{\gamma}(s)\big)=1
    \end{equation*}
    with~$\gamma(s)\in N_{t_{0}+k \Delta t}$. We also define for each point~$p\in M$ the set~$D_{p}\scrL^{k+1}$:
    \begin{equation*}
        D_{p}\scrL^{k+1}\coloneqq \big\{ \dot{\gamma}(s)\; \big| \;(I,\gamma)\in \scrL \;\textrm{and}\; \gamma(s)=p \big\}
    \end{equation*}
    
\end{enumerate}
Hence, in this way, we have constructed a timelike vector field on the Cauchy hypersurfaces separated by a time step~$\Delta t$. Furthermore, if we now take the limit~$k_{\max} \to\infty$, a global timelike vector field~$u$ is obtained.

\begin{Remark}\label{rem:shape_DxL_locally_rigid} {\em{
    The name \emph{locally rigid} for the dynamics introduced in the previous definition can be motivated as follows. Consider a foliation~$(N_{t})_{t\in\mathbb{R}}$ with initial~$N_{t_{0}}$,~$\Delta t >0$ small enough and an initial timelike vector field~$u: N_{t_{0}} \rightarrow TM$. As discussed in Remark~\ref{rem:shape_DxL}, the hypersurface~$D_{p}\scrL$ has an ellipsoidal shape for any~$p\in N_{t_{0}}$ whereas for~$q\in N_{t_{0}+\Delta t}$,~$D_{q}\scrL$ is, in general, not an ellipsoid anymore. However, if~$u$ follows the locally rigid dynamics, then~$D_{q}\scrL'$ still has an ellipsoidal shape. In other words, the dynamics given by Definition~\ref{def:locally_rigid_dynamics} is 'locally rigid', in the sense that at each point~$q\in M$ there exists a timelike vector field~$u_{q}$ and an associated~$D_{q}\scrL$ (or~$D_{q}\scrL'$) which preserves an ellipsoidal shape (see the upper part of~\ref{fig1}).
}} \hfill\QEDrem \end{Remark}

Until now the locally rigid dynamics of the regularizing vector field~$u$ has been introduced. However, the interest in this vector field is that it should describe the evolution of regularized spinor fields in~$(M,g)$ according to a dynamics that deviates slightly from the Dirac dynamics. How is the locally rigid dynamics of the spinor fields linked to the locally rigid dynamics of~$u$? This is addressed in the following definition.

\begin{Def}[Locally rigid operator]\label{def:locally_rigid_dynamics_spinors}
    Let~$(M,g)$ be a globally hyperbolic spacetime and~$(N_{t})_{t\in\mathbb{R}}$ a foliation. Furthermore, let~$u : M \rightarrow TM$ be the regularizing vector field satisfying the locally rigid dynamics
(see Definition~\ref{def:locally_rigid_dynamics}). Moreover, let~$(A_{t})_{t\in\mathbb{R}}$ be the associated family of symmetrized Hamiltonians
(as introduced in~\eqref{eq:baryogenesis_operator}) and~$\H_{t_{0}}^{\varepsilon}\coloneqq \chi_{I}(A_{t_{0}})\H_{t_{0}}$. For each~$\Delta t>0$ and iterating in~$k\in \mathbb{N}$, we introduce
the following spaces of regularized spinor fields,
\[ \H^{\varepsilon}_{t_{0}+k\Delta t}\coloneqq \big( \chi_{I}(A_{t_{0}+k\Delta t}) U^{t_{0}+k\Delta t}_{t_{0}+(k-1)\Delta t}\cdot\cdot\cdot \chi_{I}(A_{t_{0}+\Delta t})U^{t_{0}+\Delta t}_{t_{0}} \big)(\H^{\varepsilon}_{t_{0}}) \:, \]
where for any~$t_{k}<t_{k+1}$, the operator~$U_{t_{k}}^{t_{k+1}} : \H_{t_{k}} \rightarrow \H_{t_{k+1}}$ is the unitary operator that describes the Dirac evolution of the regularized spinor fields. Moreover, the {\bf{locally rigid operator}}~$V_{t_{0}}^{t} : \H^{\varepsilon}_{t_{0}}\rightarrow C^{0}(N_{t},SM)\cap \H_{t}$ is defined by
\[ 
        V_{t_{0}}^{t}\coloneqq \lim_{k_{\max}\to\infty} \chi_{I}(A_{t}) U^{t}_{t-\Delta t}\cdot\cdot\cdot \chi_{I}(A_{t_{0}+\Delta t})U^{t_{0}+\Delta t}_{t_{0}} \quad \text{with} \qquad \Delta t\coloneqq \frac{t-t_{0}}{k_{\max}} \:. \]
\end{Def}

\par By construction, the locally rigid operator~$V_{t_{0}}^{t}$ describes the locally rigid evolution of the regularized spinor fields. The adiabatic projections have the advantage of implementing deviations from the Dirac dynamics. Moreover, they guarantee that the locally rigid operator~$V_{t_{0}}^{t}$ is unitary and thus that the scalar product~$(.|.)_{t}$ is independent of the chosen Cauchy hypersurface.

\section{The rate of baryogenesis in Minkowski spacetime}\label{sec:rate_of_baryo_Minkowski}
In this section we determine the baryogenesis rate in Minkowski spacetime in the case that the regularized spinor fields follow the locally rigid dynamics. In the first place, a dynamical equation describing the locally rigid evolution of the regularizing vector field~$u$ will be derived (cf. Lemma~\ref{lem:dyn_eq}). Afterward, we will construct an explicit expression for the projection valued measure~$\chi_{I}(H_{\eta})$ associated to the Dirac Hamiltonian (cf. Lemma~\ref{lem:spectral_measure}), we will show that certain products of operators (involving~$A_{t}$ and~$\chi_{I}(H_{\eta})$) are trace-class and prove a useful distributional equation (see Lemma~\ref{lem:sokhotski}). All these results will play an important role in order to prove   Theorem~\ref{theo:baryo_rate_Minkowski}, the main result of this section, where an approximate formula for the rate of baryogenesis is derived. In particular, we show that, even though baryogenesis typically does take place if we assume a non-trivial initial regularizing vector field~$u$ that evolves according to the locally rigid dynamics, it is only a second order correction to the Dirac dynamics.

Minkowski spacetime~$(\mathbb{R}^{1,3},\eta)$ in Cartesian coordinates is given by the manifold~$\mathbb{R}^{4}$ with the metric
\begin{equation*}
    \eta=dt^{2}-dx^{2}-dy^{2}-dz^{2}\:.
\end{equation*}
Note that instead of the coordinates~$(t,x,y,z)$ we will often use~$(x^{0},x^{1},x^{2},x^{3})$. With respect to Clifford multiplication, the Dirac representation will be used. In other words, Clifford multiplication will be given by the following matrices\
\[
\gamma_{\eta}^{0}=
\begin{pmatrix}
  \textrm{Id} & 0 \\ 0 & -\textrm{Id}
\end{pmatrix}
,\hspace{1cm}
\gamma_{\eta}^{j}=
\begin{pmatrix}
  0 & \sigma^{j} \\ -\sigma^{j} & 0
\end{pmatrix} \:,
\]

where~$\sigma^{j}$ denotes the Pauli matrices:

\[
\sigma^{1}=
\begin{pmatrix}
    0 & 1 \\ 1 & 0
\end{pmatrix}
,\hspace{1cm}
\sigma^{2}=
\begin{pmatrix}
   0 & -i \\ i & 0 
\end{pmatrix}
,\hspace{1cm}
\sigma^{3}=
\begin{pmatrix}
    1 & 0 \\ 0 & -1
\end{pmatrix}
\]
Clearly, with respect to the usual inner product~$\langle\cdot|\cdot\rangle_{\mathbb{C}^{4}}$ on~$\mathbb{C}^{4}$, the previous matrices satisfy that~$(\gamma^{0})^{\ast}=\gamma^{0}$ and~$(\gamma^{\mu})^{\ast}=-\gamma^{\mu}$ for~$\mu\in\{1,2,3\}$. Moreover, with respect to the inner product~$\Sl\cdot|\cdot\Sr_{S_{p}M}=\langle\gamma_{\eta0}\cdot|\cdot\rangle_{\mathbb{C}^{4}}$ on the spin space~$S_{p}M\cong \mathbb{C}^{4}$,
a direct computation (using that~$\gamma_{\eta}^{0}\gamma_{\eta}^{\mu}=-\gamma_{\eta}^{\mu}\gamma_{\eta}^{0}$) yields that all matrices~$\gamma_{\eta}^{j}$ are symmetric with respect to this inner product (i.e.\ $(\gamma^{j})^{\ast}=\gamma^{j}$). This was already mentioned to hold for general globally hyperbolic spacetimes~$(M,g)$ in Section~\ref{sec:preliminaries}.

In the following lemma we determine an approximate evolution equation for the regularizing vector field~$u$ considering that, starting from an arbitrary initial value (i.e.\ considering a general initial vector field on the Cauchy hypersurface~$N_{t_{0}}$) it follows the locally rigid dynamics given by Definition~\ref{def:locally_rigid_dynamics}.

\begin{Lemma}\label{lem:dyn_eq}
    Let~$(N_{t})_{t\in\mathbb{R}}$ be the foliation of Minkowski spacetime given by the level sets of the
    global time function~$t$. Given an initial time~$t_0$, we consider a positive
    and smooth function~$f\in C^{\infty}(\mathbb{R}^{3},\mathbb{R}_{>0})$ on the initial Cauchy surface~$N_{t_{0}}$ (with~$f_{p}=f (x,y,z)$) and an arbitrary constant~$\lambda\geq0$.
Assume that on this Cauchy surface, the regularizing vector field is given by
\begin{equation}\label{eq:u_initially}
    u_{p}=f_{p}\nu+\lambda X_{p} \qquad \text{for all~$p\in N_{t_{0}}$,} \:
\end{equation}
where~$\nu$ is the normal vector field to the Cauchy hypersurface~$N_{t_{0}}$ and~$X_{p}\in T_{p}M$ is an arbitrary spacelike vector field.
Then, the locally rigid dynamical equation of~$u$ 
at~$p \in N_{t_0}$ is given to first order in~$\lambda$ by
    \begin{equation}\label{eq:dynamical_eq_u_Minkowski}
         \frac{du_{p}}{dt}=-\textrm{grad}_{\delta}(f_{p}^{-1})+\frac{\lambda}{f_{p}^{3}}\left(\frac{f_{p}}{3}\textrm{div}_{\delta}\left(X_{p}\right)+4X_{p}(f_{p})\right)\nu+\mathcal{O}(\lambda^{2})
    \end{equation}
    where~$\delta$ denotes the Euclidean metric in~$\mathbb{R}^{3}$.
\end{Lemma}
\begin{proof}
Before delving into the details of the proof, note that~$u_{p}$ can be rewritten in local coordinates as follows,
\begin{equation}\label{eq:u_linear_comb}
    u_{p}=u_{p,1}+u_{p,2}+u_{p,3} \hspace{1cm}\text{with} \hspace{1cm}
    \begin{cases}
\displaystyle    u_{p,1}= \frac{f_{p}}{3}\partial_{t}+\lambda X^{x}_{p}\partial_{x}\\[0.5em]
\displaystyle    u_{p,2}= \frac{f_{p}}{3}\partial_{t}+\lambda X^{y}_{p}\partial_{y}\\[0.5em]
\displaystyle    u_{p,3}= \frac{f_{p}}{3}\partial_{t}+\lambda X^{z}_{p}\partial_{z}\\
\end{cases}
\end{equation}
where we used that, in local coordinates, the following relations hold,
\[ \nu_{p}=\partial_{t} \qquad \text{and} \qquad X_{p}=X^{\alpha}_{p}\partial_{\alpha} \quad \text{with} \quad X^{\alpha}_{p}=X^{\alpha}_{p}(x,y,z) \:. \]
Then, by linearity, computing~$\frac{du_{p}}{dt}$ reduces to determining dynamical equations of the form~$\frac{du_{p,\alpha}}{dt}$, with~$\alpha\in\{1,2,3\}$. Computing the latter is much easier as now each~$u_{p,\alpha}$ is the sum of a timelike vector field and a spacelike vector field which points exclusively in one direction. Moreover, we want to determine the dynamical equation linearly in~$\lambda$. For this reason, without loss of generality, in this proof it suffices to consider the case that~$X_{p}$ only points in one direction, i.e.\ we consider that, in local coordinates, it holds that
\[ u_{p}=f_{p}\partial_{t}+\lambda h_{p}\partial_{\mu} \qquad \text{for all~$p\in N_{t_{0}}$} \]
with~$\mu\in\{1,2,3\}$, and where~$h\in C^{\infty}(\mathbb{R}^{3},\mathbb{R}_{>0})$ is a function of all three coordinates~$(x,y,z)$ . At the end of the proof, when rewriting the coordinate-dependent dynamical equation for~$u_{p}$ in a more abstract and geometric way, we will, however, take into account that the spacelike vector field can point in an arbitrary direction (as in expression~\eqref{eq:u_linear_comb}). However, it is important to note that the previous argument would not work if we wanted to obtain the expansion of the dynamical equation to second (or higher) order in~$\lambda$ (i.e.\ terms which are not linear in~$\lambda$ anymore).

\par In order to compute the locally rigid dynamical equation of~$u$, we need to use some of the formulas derived in the proof of~\cite[Theorem A.2]{baryogenesis}. For the sake of completeness, we begin by presenting the first part of their proof until we reach the expressions needed for the derivation of the dynamical equation.

In the first part of the proof, the set~$\scrL$ and~$D_{p}\scrL$ (for an arbitrary~$p\in M$) are determined explicitly. Let~$(N_{t})_{t\in\mathbb{R}}$ be the foliation of Minkowski spacetime given by the level sets of the global time function~$t$. Following Definition~\ref{def:DxL},~$\scrL$ is the set of maximal null geodesics~$\gamma : I \rightarrow \mathbb{R}^{1,3}$ (together with the interval~$I\subset \mathbb{R}$) that satisfy the relation
\begin{align}\label{eq:LemmaDxL}
    g(u_{p},\xi_{p}(n))=1
\end{align}
whenever~$p=\gamma(s)\in N_{t_{0}}$, and where~$\xi_{p}(n)$ denotes the tangent vector of~$\gamma$ at the point~$p=\gamma(s)$. The reason that we denote the tangent vector by~$\xi_{p}(n)$ is due to the fact that any null vector~$\dot{\gamma}(s)$ can also be written in the computationally convenient form (see~\cite[Appendix A]{dgc})
\begin{align}\label{eq:def_xip}
    \dot{\gamma}(s)= b_{p}(n)\zeta(n)=: \xi_{p}(n).
\end{align}
where~$n\in \mathbb{S}^{2}$,~$\zeta(n)=(1,n)$ and~$b_{p} : \mathbb{S}^{2}\rightarrow (0,\infty)$. Note that equations~\eqref{eq:LemmaDxL} and~\eqref{eq:def_xip} imply that for any~$p\in N_{t_{0}}$, the function~$b_{p}(n)$ must be given by
\begin{align}\label{eq:fp_initially}
    b_{p}(n)=\frac{1}{g(u_{p},\zeta(n))}\;.
\end{align}
This equation fully determines the set~$\scrL$. Moreover, at any point~$q\in M$ (not necessarily~$q\in N_{t_{0}}$),~$D_{q}\scrL$ can be written as the following set:
\[ D_{q}\scrL\coloneqq \{\xi_{q}(n)=\dot{\gamma}(s)\; | \;(I,\gamma)\in \scrL \;\textrm{and}\; \gamma(s)=q\} \:. \]
In particular, the measure on~$D_{q}\scrL$ presents the simple form (see~\cite[equation (A.3)]{dgc}):
\[ d\mu_{q}= b_{q}(n)^{2} \:d\mu_{\mathbb{S}^{2}} \:, \]
where~$d\mu_{\mathbb{S}^{2}}$ is the Lebesgue measure on the sphere~$\mathbb{S}^{2}$.
Whereas~\eqref{eq:fp_initially} fixes the explicit value of~$b_{p}(n)$ (and thus~$\xi_{p}(n)$) for all~$p\in N_{t_{0}}$, how does~$b_{q}(n)$ for a point~$q\not\in N_{t_{0}}$ look like? Note that this is a very important question in order to compute the dynamical equation of~$u$. The value of~$b_{q}(n)$ and thus~$D_{q}\scrL$ are determined by the Dirac dynamics of the spinor fields. The Dirac dynamics is easily implemented in Minkowski spacetime because null geodesics are given by straight lines. Let~$q\in N_{t'}$ and~$p'\in N_{t_{0}}$ be two points with coordinates~$(t',0,0,0)$ and~$(t_{0},-t' n^{1},-t' n^{2},-t' n^{3})$ and~$t'>t_{0}$. Then, the Dirac dynamics is simply implemented by demanding that 
\[ b_{q}(n) = b_{p'}(n) \:. \]
This finishes the first part of the proof in which the main objects that will be used in the computations are presented.

\par In the second part of the proof we determine the differential equation that governs the locally rigid dynamics of the regularizing vector field~$u$. Without loss of generality we choose~$t_{0}=0$. Moreover, consider the points~$p, q, p'\in M$ with coordinates~$(0,0,0,0),(\Delta t,0,0,0)$ and~$(0,-\Delta t n^{1},-\Delta t n^{2},-\Delta t n^{3})$, respectively (with~$n\in \mathbb{S}^{2}$ a unit vector). 
    \par In the first place, for a fixed unit vector~$n\in\mathbb{S}^{2}$ the linear Taylor expansion in~$\Delta t$ of~$b_{p'}(n)=b_{(0,-\Delta t n)}(n)$ is
    \begin{equation}\label{eq:exp_fp}
        b_{p'}(n)=b_{p}(n)+\left.\frac{\partial b_{(0,-\Delta t n)}}{\partial \Delta t}\right|_{\Delta t=0}\Delta t + \mathcal{O}\big( (\Delta t)^{2} \big)\;,
    \end{equation}
    where
    \begin{align}
        b_{p}(n)&=\frac{1}{g(u_{p},\zeta(n))}\label{eq:fp}\\
        \left.\frac{\partial b_{(0,-\Delta t n)}}{\partial \Delta t}\right|_{\Delta t=0}&=-\frac{1}{[g(u_{p},\zeta(n))]^{2}} \left(\left.\frac{\partial}{\partial \Delta t}g(u_{(0,-\Delta t n)},\zeta(n))\right)\right|_{\Delta t=0}\nonumber\\
        &=\frac{1}{[g(u_{p},\zeta(n))]^{2}} n^{\mu}g(\partial_{\mu}u_{p},\zeta(n))\nonumber\\
        &=\frac{1}{|g(u_{p},\zeta(n))|^{2}} \; \zeta^i \:\pi(\nu^\perp)^j_i\: g(\zeta(n),\partial_j u_{p} ) \label{eq:partial_fp}\;.
    \end{align}
    Here~$\zeta^{j}$ denotes the components of the vector~$\zeta(n)$ and for an arbitrary timelike vector~$v$ we denote the projections to the span of~$v$ and to its orthogonal complement by
\[ \pi(v)^i_j := \frac{v^i v_j}{v^2} \qquad \text{and} \qquad
\pi \big( v^\perp \big)^i_j := \delta^i_j -\frac{v^i v_j}{v^2} \:. \]
Substituting~\eqref{eq:fp} and~\eqref{eq:partial_fp} into~\eqref{eq:exp_fp}, we get the following expression,
    \[ b_{p'}(n) = \frac{1}{g(u_{p},\zeta(n))} 
+ \frac{\Delta t}{[g(u_{p},\zeta(n))]^2} \; \zeta^i \:\pi(\nu^\perp)^j_i\: \big\la \zeta(n),\, \partial_j u_{p} \big\ra 
+ \O\big( (\Delta t)^2 \big) \:. \]
Next, using that
\[ b_{p}(n) = \frac{1}{g(u_{p},\zeta(n))} \qquad \text{and} \qquad
\xi_{p}^{j}=b_{p}(n)\zeta^{j}  \,\:, \]
we obtain
\begin{align*}
b_{p'}(n) &= b_{p}(n) \,\bigg( 1
+ \Delta t\zeta^ib_{p}(n) \:\pi(\nu^\perp)^j_i\: g(\zeta(n),\, \partial_j u_{p} ) \bigg)
+ \O\big( (\Delta t)^2 \big) \\
&=b_{p}(n) \,\bigg( 1
+ \frac{\Delta t}{b_{p}(n)} \xi_{p}^{i} \; \pi(\nu^\perp)^j_i\: g(\xi_{p}(n),\, \partial_j u_{p} ) \bigg)
+ \O\big( (\Delta t)^2 \big) \\
&= b_{p}(n) \,\bigg(1 + \frac{\Delta t}{g(\xi_{p}(n), \nu)}\; \xi_{p}^i \:\pi(\nu^\perp)^j_i\:  g (\xi_{p},\, \partial_j u_{p}) \bigg)
+ \O\big( (\Delta t)^2 \big) \:,
\end{align*}
where in the last step we used that~$b_{p}(n)=g(\nu,\xi_{p}(n))$. As discussed in the first part of the proof, we now implement the Dirac dynamics by demanding that~$b_{q}(n)=b_{p'}(n)$. We thus obtain
    \begin{align}
        &d\mu_{q}(\xi_{q})=b^{2}_{q}(n)d\mu_{\mathbb{S}^{2}}(n)=b^{2}_{p'}(n)d\mu_{\mathbb{S}^{2}}(n)\nonumber\\
        &=b_{p}^{2}(n)\left[1+2\Delta t\Cintegralwithn+\mathcal{O}((\Delta t)^{2})\right]d\mu_{\mathbb{S}^{2}}(n)\label{eq:dmuq}\\
        &\xi_{q}d\mu_{q}(\xi_{q}) =\xi_{p'}d\mu_{p'}(\xi_{p'})=b^{3}_{p'}(n)\zeta(n)d\mu_{\mathbb{S}^{2}}(n)\nonumber\\
        &=b_{p}^{3}(n)\left[1+3\Delta t\Cintegralwithn+\mathcal{O}((\Delta t)^{2})\right]\xi_{p}(n)d\mu_{\mathbb{S}^{2}}(n)\label{eq:xiq_dmuq} .
    \end{align}
    Note that in~\eqref{eq:dmuq} and~\eqref{eq:xiq_dmuq} only spatial derivatives of~$u_{p}$ come into play because~$\pi(\nu^{\perp})^{0}_{j}=0$ for all~$j\in\{1,2,3\}$. From now on, for notational simplicity, we drop explicit reference to the normal vector~$n$.
    \par By the chain rule we get:
    \begin{align}
        \overline{\xi}_{p}&\coloneqq\frac{1}{\mu_{p}(D_{p}\scrL)}\int_{D_{p}\scrL}\xi_{p} d\mu_{p}(\xi_{p})\nonumber\\
        \frac{d\overline{\xi}_{p}}{dt}&=-\frac{\overline{\xi}_{p}}{\mu_{p}(D_{p}\scrL)}\frac{d}{dt}(\mu_{p}(D_{p}\scrL))+\frac{1}{\mu_{p}(D_{p}\scrL)}\int_{D_{p}\scrL}\frac{d}{dt}(\xi_{p} d\mu_{p}(\xi_{p}))\label{eq:chainrule1}
    \end{align}
    We compute each term individually using~\eqref{eq:dmuq} and~\eqref{eq:xiq_dmuq} and combine them with expression~\eqref{eq:chainrule1} in order to obtain
    \begin{align*}
       &\frac{d}{dt}(\mu_{p}(D_{p}\scrL))=\lim_{\Delta t\to 0}\frac{\mu_{p+\Delta t}(D_{p+\Delta t}\scrL)-\mu_{p}(D_{p}\scrL)}{\Delta t}=\lim_{\Delta t\to 0}\frac{\mu_{q}(D_{q}\scrL)-\mu_{p}(D_{p}\scrL)}{\Delta t}\\
       &=\lim_{\Delta t\to 0}\frac{1}{\Delta t}\left[\int_{\mathbb{S}^{2}}b_{p}^{2}\left(1+2\Delta t \Cintegral+\mathcal{O}((\Delta t)^{2})\right)d\mu_{\mathbb{S}^{2}}-\int_{\mathbb{S}^{2}}\!\!b_{p}^{2}d\mu_{\mathbb{S}^{2}}\right]\\
       &=\int_{D_{p}\scrL}2\Cintegral d\mu_{p}(\xi_{p})\\
       &\frac{d}{dt}(\xi_{p} d\mu_{p}(\xi_{p}))=\lim_{\Delta t\to 0}\frac{\xi_{q} d\mu_{q}(\xi_{q})-\xi_{p} d\mu_{p}(\xi_{p})}{\Delta t}\\
       &=\lim_{\Delta t\to 0}\frac{1}{\Delta t}\left[\int_{\mathbb{S}^{2}}b_{p}^{3}\left(1+3\Delta t \Cintegral+\mathcal{O}((\Delta t)^{2})\right)d\mu_{\mathbb{S}^{2}}-\int_{\mathbb{S}^{2}} \!\! b_{p}^{3}d\mu_{\mathbb{S}^{2}}\right]\\
       &=\int_{D_{p}\scrL}3\Cintegral \xi_{p}d\mu_{p}(\xi_{p})\\
       &\frac{d\overline{\xi}_{p}}{dt}=\int_{D_{p}\scrL}\Cintegral \left[3\xi_{p}-2\frac{u_{p}}{u_{p}^{2}}\right]d\mu_{p}(\xi_{p})|\:,
    \end{align*}
    where we also used that~$\overline{\xi}_{p}=\frac{u_{p}}{|u_{p}|^{2}}$. 
    
    For ease in notation, from now on we omit the subscript~$p$, which denotes the point at which the vector fields are evaluated.
    As~$u=\frac{\overline{\xi}}{|\overline{\xi}|^{2}}$, applying the chain rule and using the previous computation, we get the following dynamical equation for~$u$,
    \begin{align*}
\frac{du}{dt} &=\frac{1}{|\overline{\xi}|^{2}}\frac{d\overline{\xi}}{dt}-2\frac{\overline{\xi}}{|\overline{\xi}|^{4}}g\left(\overline{\xi},\frac{d \overline{\xi}}{dt}\right) \\
&= u^{2}\frac{d\overline{\xi}}{dt}-2\frac{u}{|\overline{\xi}|^{2}}g\left(\overline{\xi},\frac{d \overline{\xi}}{dt}\right)=u^{2}\left[\frac{d\overline{\xi}}{dt}-2ug\left(\overline{\xi},\frac{d\overline{\xi}}{dt}\right)\right]\\
        &=\frac{1}{\mu_{p}(D_{p}\scrL)}\int_{D_p\scrL} \Cintegralwithoutp\left[3u^{2}\xi+2u-6uu^{i}\xi_{i}\right]d\mu(\xi) \:,
    \end{align*}
    where we used that
\[ u^{2}=\frac{1}{\overline{\xi}^{2}} \qquad \text{and} \qquad
2g(\overline{\xi},3\xi-2\overline{\xi})=6\overline{\xi}^{i}\xi_{i}-4 \overline{\xi}^{2}=\frac{1}{u^2}\left[6u^{i}\xi_{i}-4\right] \:. \]
This is a system of quasi-linear partial differential equations of first order which can be rewritten for each component of~$u$ as follows
    \begin{align}\label{eq:dif_eq_u}
        & \frac{du^{l}}{dt}=\Constantone\left[3 u^{i}u_{i} (I_{3})^{knl}+2u^{l}(I_{2})^{kn}-6u^{l}u^{i}g_{si}(I_{3})^{nks}\right]\partial_{j}u^{m}\;,
    \end{align}
    where the integrals~$I_{3}$ and~$I_{2}$ are given by
    \begin{align}\label{eq:dynamical_eq_u}
        &(I_{3})^{knl}=\volDpL\intDpL \frac{1}{g(\xi,\nu)}\xi^{k}\xi^{n}\xi^{l}d\mu(\xi)\\
        &(I_{2})^{kn}=\volDpL\intDpL \frac{1}{g(\xi,\nu)}\xi^{k}\xi^{n}d\mu(\xi) \:.
    \end{align}
The expressions~\eqref{eq:dif_eq_u} are the differential equations that determine the locally rigid dynamical equation of~$u$. This finishes the second part of the proof. 
\par In the last part of the proof, the integrals~$(I_{2})^{kn}$ and~$(I_{3})^{knl}$ will be computed and a simpler locally rigid dynamical equation for~$u$ is obtained by expanding~\eqref{eq:dif_eq_u} in powers of~$\lambda$. It is only at this stage that the particular form considered for the regularizing vector field~$u$ on the initial Cauchy hypersurface will come into play.
\par The integrals~$(I_{3})^{knl}$ and~$(I_{3})^{knl}$ are easier to compute after applying a (proper orthochronous) Lorentz transformation~$\Lambda\in \SO^{+}(3,1)$ such that, in local coordinates,~$u'=\Lambda u\overset{!}{=}|u|(1,0)$. In this reference frame,~$f'(n)=\frac{1}{|u|}$ and thus~$\xi'^{j}=\frac{1}{|u|}\zeta^{j}$. Hence, the integrals simplify to
    \begin{align}\label{eq:boosted_integrals}
        &(I'_{3})^{knl}=\frac{1}{|u|^{4}}\int_{\mathbb{S}^{2}}\frac{1}{g(\zeta,\nu)}\zeta^{k}\zeta^{n}\zeta^{l}d\mu_{\mathbb{S}^{2}}\\
        &(I'_{2})^{kn}=\frac{1}{|u|^{3}}\int_{\mathbb{S}^{2}}\frac{1}{g(\zeta,\nu)}\zeta^{k}\zeta^{n}d\mu_{\mathbb{S}^{2}}\:,
    \end{align}
where the measure on the sphere is~$d\mu_{\mathbb{S}^{2}}=d\varphi dv$ with~$v=\cos{\theta}
\in (-1,1)$ and~$\varphi \in [0, 2 \pi)$. Moreover, $\zeta(n)=(1,n)=(1,\sqrt{1-v^{2}}\cos{\varphi},\sqrt{1-v^{2}}\sin{\varphi},v)$. The original integrals are obtained by the transformations
\[ (I_{3})^{knl}=(\Lambda^{-1})^{k}_{i}(\Lambda^{-1})^{n}_{j}(\Lambda^{-1})^{l}_{k}(I_{3}')^{ijk}
\qquad \text{and} \qquad (I_{2})^{kn}=(\Lambda^{-1})^{k}_{i}(\Lambda^{-1})^{n}_{j}(I'_{2})^{ij} \:. \]
The integrals given by~\eqref{eq:boosted_integrals} are straightforward to compute\footnote{This
computation was carried out with the help of computer algebra.
The corresponding {\textsf{Wolfram Mathematica}} worksheet is included as an ancillary file to the arXiv submission of this paper.}. In order to obtain the dynamical equation of~$u$ we study separately the following cases:
    \begin{enumerate}[leftmargin=2em]
        \item \underline{$u=(f_p \;,\;\lambda h\;,\;0\;,\;0)$:} \\ In the first case, the series expansion of equation~\eqref{eq:dynamical_eq_u} to first order in~$\lambda$ yields
\[ \frac{du}{dt}=\bigg(\frac{4\lambda h}{f^{3}}\partial_{x}(f)+\frac{\lambda}{3f^{2}}\partial_{x}(h),\;\frac{\partial_{x}f}{f^{2}},\;\frac{\partial_{y}f}{f^{2}},\;\frac{\partial_{z}f}{f^{2}} \bigg)
+\mathcal{O} \big(\lambda^{2} \big)\:. \]
        \item \underline{$u=(f\;,\;0\;,\;\lambda h\;,\;0)$:} \\
        In the second case, the series expansion of equation~\eqref{eq:dynamical_eq_u} to first order in~$\lambda$ yields
        \begin{equation*}
       \frac{du}{dt}=\bigg(\frac{4\lambda h}{f^{3}}\partial_{y}(f)+\frac{\lambda}{3f^{2}}\partial_{y}(h),\;\frac{\partial_{x}f}{f^{2}},\;\frac{\partial_{y}f}{f^{2}},\;\frac{\partial_{z}f}{f^{2}}\bigg) + \mathcal{O}\big(\lambda^{2} \big)\:.
        \end{equation*}
        \item \underline{$u=(f\;,\;0\;,\;0\;,\;\lambda h)$:}\\
        Finally, the series expansion of equation~\eqref{eq:dynamical_eq_u} to first order in~$\lambda$ yields
        \begin{equation*}
       \frac{du}{dt}=\bigg(\frac{4\lambda h}{f^{3}}\partial_{z}(f)+\frac{\lambda}{3f^{2}}\partial_{z}(h),\;\frac{\partial_{x}f}{f^{2}},\;\frac{\partial_{y}f}{f^{2}},\;\frac{\partial_{z}f}{f^{2}}\bigg) +\mathcal{O} \big( \lambda^{2} \big)\:.
        \end{equation*}
       \end{enumerate}
Hence, for a general initial timelike vector field~$u= f_{p} \partial_{t}+\lambda X_{p}$ (where~$X_{p}$ is an arbitrary spacelike vector field, recall the discussion at the beginning of this proof), the locally rigid dynamical equation becomes
\begin{align*}
         \frac{du_{p}}{dt}&=\left(4\frac{\partial_{\mu}(f_{p})}{f_{p}^{3}}h_{p}+\frac{\partial_{\mu}(h_{p})}{3f_{p}^{2}}\right)\lambda\partial_{t}+\frac{\partial_{x}f_{p}}{f^{2}_{p}}\partial_{x}+\frac{\partial_{y}f_{p}}{f_{p}^{2}}\partial_{y}+\frac{\partial_{z}f_{p}}{f_{p}^{2}}\partial_{z}+\mathcal{O}(\lambda^{2})\\
         &=-\textrm{grad}_{\delta}(f_{p}^{-1})+\frac{\lambda}{f_{p}^{3}}\left(\frac{f_{p}}{3}\textrm{div}_{\delta}\left(X_{p}\right)+4X_{p}(f_{p})\right)\nu+\mathcal{O}(\lambda^{2})\;.
\end{align*}
This concludes the proof.
\end{proof}
\par Note that, while the parameter~$\lambda$ scales the spatial component of the initial regularizing vector field~$u$ in the direction of the spacelike vector~$X$, the function~$f$ scales the normal component of~$u$ to the initial Cauchy hypersurface at each point.
\begin{Remark} {\em{
    A surprising feature of the previous evolution equation is that, even in the simplest case in which~$\lambda=0$, the locally rigid dynamics of~$u$ is still non trivial, i.e.\
\[ \frac{du_{p}}{dt}=-\textrm{grad}_{\delta}(f_{p}^{-1}) \:. \]
    In other words, even if initially the regularizing vector field is proportional to~$\partial_{t}$, the locally rigid evolution of~$u$ gives rise to a non-zero spatial component of 
the regularizing vector field: unless~$u|_{N_{t_{0}}}$ is constant (i.e.\ $\partial_{\mu}f_{p}=0$ for all~$\mu\in\{1,2,3\}$ and~$p\in N_{t_{0}}$), at a later time,~$u$ is not parallel to~$\partial_{t}$ anymore. }} \hfill\QEDrem
\end{Remark}
Note that in the following lemmata we will assume that, on the initial Cauchy hypersurface, the regularizing vector field is given by
\begin{equation}\label{eq:u_initially_simpler}
    u_{p}=(1+\lambda\Tilde{f}_{p})\nu+\lambda X_{p} \qquad \text{for all~$p\in N_{t_{0}}$}
\end{equation}
with~$\Tilde{f}$ a smooth positive function of the spatial coordinates. Note that this ansatz for the regularizing vector field agrees with that of the previous lemma (cf.~\eqref{eq:u_initially}) if we choose~$f= 1 + \lambda \Tilde{f}$. The main advantage of expression~\eqref{eq:u_initially_simpler} is that it makes a perturbative analysis (and the involved computations) simpler. In particular, by perturbing around~$\lambda=0$ we recover a well known scenario: if the foliation~$(N_{t})_{t\in\mathbb{R}}$ corresponds to the level sets of the global time function~$t$ and~$\lambda=0$, the regularizing vector field~$u$ agrees with~$\partial_{t}$ and the symmetrized Hamiltonian~$A_{t}$ with the Dirac Hamiltonian~$H_{\eta}$. Hence, perturbations around~$\lambda=0$ correspond to corrections to the Dirac dynamics. 
\par In the following lemma we construct the projection valued measure associated to the Dirac Hamiltonian. This operator will play an important role in the formula for the rate of baryogenesis due to the locally rigid evolution of the spinor fields.

\begin{Lemma}\label{lem:spectral_measure}
Let~$\varepsilon>0$ and consider the interval~$I\coloneqq (-1/\varepsilon,\omega)$ with~$\omega\leq-m$ and let~$\chi_{I}(H_{\eta})$ denote the projection valued measure associated to the Dirac Hamiltonian~$H_{\eta}$. Then, $\chi_{I}(H_\eta)$ is given by the following integral operator: 
\beq \label{chiIH}
\big( \chi_{I}(H_\eta) \psi \big)(x) = -2\pi \int_{\R^3} 
P^{\varepsilon,\omega}(x,y)\:\gamma_{\eta}^{0}\: \psi(y)\: d^3y \:,
\eeq
where ~$\psi\in C^{\infty}_{\textrm{sc}}(M,SM)\cap \H_{m}$, ~$p\in M$ is a point with coordinates~$x=(t,x^{\mu})$, $y=(0,y^{\mu})$ and ~$P^{\varepsilon,\omega}(x,y)$ is again the distributional kernel~\eqref{Pepsdef_omega}. Moreover, $\chi_{I}(H_{\eta})$ can be written as follows:
\begin{equation}\label{eq:pvmDiracHamiltonian}
    \chi_{I}(H_{\eta})=\int_{-1/\varepsilon}^{\omega}F_{\omega}(H_{\eta})d\omega
\end{equation}
where~$F_{\omega}(H_{\eta})=\frac{d}{d\omega}\big(\chi_{(-\frac{1}{\varepsilon},\omega)}(H_{\eta}) \big)$ is an integral operator whose integral kernel~$F_{\omega}(x,y)$ satisfies the following expressions:
\begin{align}
    &F_{\omega}(x,x)=-\frac{1}{(2\pi)^{2}}\Big((\omega+\gamma_{\eta0}m)\sqrt{\omega^{2}-m^{2}}\Big)\Theta(1+\varepsilon\omega)\label{eq:F_omega}\\
    &(\partial_{\mu}F_{\omega}(x,y))|_{y=x} =\frac{i}{3(2\pi)^{2}}\Big(\gamma_{\eta\mu}\gamma_{\eta0}(\omega^{2}-m^{2})^{3/2}\Big)\Theta(1+\varepsilon\omega)\label{eq:partial_F_omega}
\end{align}
\end{Lemma}
\begin{proof}
We start by proving expression~\eqref{chiIH}. In particular, we will explicitly construct the functional calculus of the self-adjoint operator 
\begin{equation*}
    H_{\eta} : C^{\infty}_{\textrm{sc}}(M,SM)\cap\H_{m}\rightarrow C^{\infty}_{\textrm{sc}}(M,SM)\cap \H_{m}
\end{equation*}
and afterward apply it to the characteristic function~$\chi_{I}$. It follows then that~$\chi_{I}(H_{\eta})$ is a projection valued measure.

Consider the algebra~$\mathcal{A}$ of bounded Borel functions from the spectrum~$\sigma(H_{\eta})$ to~$\mathbb{C}$. For a given function~$g\in\mathcal{A}$ and using coordinates~$x=(t,x^{\mu}), y=(0,y^{\mu})$ we define the following integral operator 
\[ \big( g(H_\eta) \psi \big)(x) = \int_{\R^3} g(H_\eta)(x,y)\: \psi(y)\:d^3y \;,\] 
with integral kernel
\begin{align}
    g(H_\eta)(x, y) &= 2 \pi 
\int \frac{d^{4}k} {(2\pi)^{4}}\: g(k^0) \:\big( \gamma_{\eta j}k^{j}+ m \big)\: \delta(k^{2}-m^2)\label{eq:func_calc_definition}\:
\epsilon(k^{0})\: e^{-ik(x-y)}\:\gamma_{\eta}^{0} \\
    &=: 2\pi Q_{g}(x,y)\gamma^{0}_{\eta}\:,\nonumber
\end{align}
where, as in Remark~\ref{ex:hard_cut_off}, $k(x-y)$ is short notation for the (Minkowski) inner product between~$k$ and~$x-y$. In the rest of the proof, we proceed to show that the map 
\begin{equation*}
    \Phi : \mathcal{A}\rightarrow \mathcal{L}(C^{\infty}_{\textrm{sc}}(M,SM)\cap \H_{m}), \quad g \mapsto g(H_{\eta})\;,
\end{equation*}
defines the functional calculus associated to the self-adjoint operator~$H_{\eta}$. Note that in the previous formula~$d^{3}y=dy^{1}dy^{2}dy^{3}$ denotes the Lebesgue measure on the spatial slice~$N_{t}\cong \mathbb{R}^{3}$.
\begin{itemize}[leftmargin=1.5em]
    \item The map~$\Phi$ is an algebra-homomorphism, i.e.\ for all~$\alpha\in\mathbb{C}$ and functions~$g,h\in\mathcal{A}$ it holds that
    \begin{align}
    & (g+\alpha h)(H_\eta) = g(H_{\eta})+\alpha h(H_{\eta}) \label{eq:func_calc_additive}\\ 
    & g(H_\eta)\, h(H_\eta) = (gh)(H_\eta) \label{eq:func_calc_multiplic}
    \end{align}
    The first expression follows directly from linearity of the integral, whereas the second one requires a longer computation. Consider points with coordinates~$x=(t,x^{\mu}), y=(t_{0},y^{\mu}), z=(t_{1},z^{\mu})$. For simplicity, we set~$t_{0}=0$. Then, for a spinor field~$\psi$ it holds that
    \begin{align*}
        \big(g(H_\eta)(h(H_{\eta})\psi )\big)(x)&=2\pi \int_{\mathbb{R}^{3}}Q_{g}(x,y)\gamma^{0}_{\eta} (h(H_{\eta})\psi)(y)d^{3}y\\\
        &=(2\pi)^{2}\int_{\mathbb{R}^{3}}\int_{\mathbb{R}^{3}}Q_{g}(x,y)\gamma^{0}_{\eta}Q_{h}(y,z)\gamma^{0}_{\eta}\psi(z)d^{3}yd^{3}z\\
        \big( (gh)(H_{\eta})\psi \big)(x) &= 2\pi \int_{\mathbb{R}^{3}} Q_{gh}(x,z)\gamma^{0}_{\eta}\psi(z)d^{3}z
    \end{align*}
    Hence, equation~\eqref{eq:func_calc_multiplic} only holds provided the following expression is satisfied:
    \begin{equation}\label{eq:Qgh}
        \int_{\mathbb{R}^{3}}Q_{g}(x,y)\gamma_{\eta}^{0}Q_{h}(y,z)d^{3}y= \frac{1}{2\pi} Q_{gh}(x,z)
    \end{equation}
    We now proceed to show that the previous equation holds, inspired by the proof of expression (1.2.24) in~\cite{cfs}. We start by rewriting the left-hand side of equation~\eqref{eq:Qgh} as follows
    \begin{align}
        &\int_{\mathbb{R}^{3}}Q_{g}(x,y)\gamma_{\eta}^{0}Q_{h}(y,z)d^{3}y\nonumber=\int_{\mathbb{R}^{3}}d^{3}y e^{iy(k-k')}\int \frac{d^{4}k}{(2\pi)^{4}}e^{-ikx}\int\frac{d^{4}k'}{(2\pi)^{4}}e^{ik'z}\widetilde{Q}_{g}(k)\gamma_{\eta}^{0}\widetilde{Q}_{h}(k')\nonumber\\
        &=\int \frac{d^{4}k}{2\pi}\delta^{3}(k^{\mu}-k'^{\mu})e^{-ikx}\int\frac{d^{4}k'}{(2\pi)^{4}}e^{ik'z}\widetilde{Q}_{g}(k)\gamma_{\eta}^{0}\widetilde{Q}_{h}(k')\nonumber\\
        &= \int \frac{d^{4}k}{(2\pi)^{4}}\int\frac{dk'^{0}}{2\pi}\Big[e^{ik'z-ikx}\widetilde{Q}_{g}(k)\gamma_{\eta}^{0}\widetilde{Q}_{h}(k')\Big]\bigg|_{k'=(k'^{0},k^{\mu})}\:, \label{eq:QkQq}
    \end{align}
    where
    \begin{align*}
        &\widetilde{Q}_{g}(k)\coloneqq g(k^0)\big( \gamma_{\eta j}k+ m \big) \delta(k^{2}-m^2)\epsilon(k^{0})\\
        &\widetilde{Q}_{h}(k')\coloneqq h(k'^0)\big( \gamma_{\eta j}k'+ m \big) \delta(k'^{2}-m^2)\epsilon(k'^{0})
    \end{align*}
    The computation in~\cite[Lemma 1.2.8]{cfs} shows that the following identity holds 
    \begin{align*}
       &\big( \gamma_{\eta j}k^{j}+ m \big)\gamma_{\eta0}\big( \gamma_{\eta j}k'^{j}+ m \big)\delta(k^{2}-m^2)\delta(k'^{2}-m^2)\\
       &=\epsilon(k^{0})\big( \gamma_{\eta j}k^{j}+ m \big)\delta(k'^{0}-k^{0}) \delta(k^{2}-m^{2})
    \end{align*}
    Inserting the previous expression in equation~\eqref{eq:QkQq} with~$k'=(k'^{0},k^{\mu})$, we obtain
\begin{align*}
     &\int_{\mathbb{R}^{3}}Q_{g}(x,y)\gamma_{\eta}^{0}Q_{h}(y,z)d^{3}y\nonumber\\
     &=\int \frac{d^{4}k}{(2\pi)^{4}}\int e^{i(k'z-kx)}\frac{dk'^{0}}{2\pi}\big( \gamma_{\eta j}k^{j}+ m \big)\delta(k'^{0}-k^{0}) \delta(k^{2}-m^{2})\epsilon(k'^{0})g(k^0)h(k'^0)\nonumber\\
     &=\frac{1}{2\pi}\int \frac{d^{4}k}{(2\pi)^{4}}e^{-ik(x-z)}\big( \gamma_{\eta j}k^{j}+ m \big) \delta(k^{2}-m^{2})\epsilon(k^{0})(gh)(k^0)=\frac{1}{2\pi}Q_{gh}(x,z)\nonumber
\end{align*}
    \item Moreover, we now show that for the identity map~$g\equiv\textrm{Id} : \mathbb{R}\to \mathbb{R} , \lambda \mapsto \lambda$, it holds that~$\textrm{id}(H_{\eta})=H_{\eta}$. As discussed in Remark~\ref{rem:fermionic_proj}, a solution to the Cauchy problem of the Dirac equation can be represented (see~\cite[Theorem 13.4.2]{intro} for the proof in Minkowski spacetime) at any point~$p$ with coordinates~$(x^{j})_{j=0,\ldots,3}$ by the following formula:
    \begin{equation*}
        \psi(x)=2\pi\int k_{m}(x, y) \gamma_{\eta0}\psi(y)d^{3}y
    \end{equation*}
    For simplicity, we set~$x=(t,x^{\mu})$ and~$y=(0, y^{\mu})$. Furthermore, in Minkowski spacetime the integral kernel~$k_{m}(x, y)$ is (cf.~\cite[equations (2.1.10), (2.1.13) and ((2.1.14))]{cfs} and~\cite[Section 16]{intro}):
    \begin{equation}\label{eq:fund_solution_Dirac_Minkowski}
        k_{m}(x, y)=\int \frac{d^{4}k} {(2\pi)^{4}}\:\:\big( \gamma_{\eta j}k^{j}+ m \big)\: \delta(k^{2}-m^2)\:
\epsilon(k^{0})\: e^{-ik(x-y)} \;.
    \end{equation}
    Hence, for~$g \equiv \textrm{id}$, we recover the following integral kernel and corresponding integral operator
    \begin{align*}
        & \textrm{id}(H_{\eta})(x,y) = 2 \pi 
\int \frac{d^{4}k} {(2\pi)^{4}}\: k^0 \:\big( \gamma_{\eta j}k^{j}+ m \big)\: \delta(k^{2}-m^2)\:\epsilon(k^{0})\: e^{-ik(x-y)}\:\gamma_{\eta}^{0}\\
        & =2 \pi i 
\int \frac{d^{4}k} {(2\pi)^{4}}\: \:\big( \gamma_{\eta j}k^{j}+ m \big)\: \delta(k^{2}-m^2)\:\epsilon(k^{0})\: \partial_{t}\big(e^{-ik(x-y)}\big)\:\gamma_{\eta}^{0}= 2\pi i\big(\partial_{t}k_{m}(x, y)\big) \\
&(\textrm{id}(H_{\eta})\psi)(x)= \int_{\R^3} \textrm{id}(H_{\eta})(x,y)\: \psi(y)\:d^3y=2\pi i\int_{\R^3} (\partial_{t}k_{m}(x, y))\: \psi(y)\:d^3y \\
&=i\partial_{t}\psi(x)=(H_{\eta}\psi)(x)\:,
    \end{align*}
    where in the last step we used that~$H_{\eta}=i\partial_{t}$ (cf. Remark~\ref{rem:Hamiltonian}). On the other hand, note that it almost directly follows from equation~\eqref{eq:fund_solution_Dirac_Minkowski} that for the constant function~$g : \mathbb{R}\to \mathbb{R}, \lambda \mapsto 1$ the associated integral operator satisfies that~$g(H_{\eta})=\textrm{id}$.

    \item Finally, we also show that for any~$g$ it holds that
    \begin{equation*}
        \overline{g}(H_{\eta})=g(H_{\eta})^{\ast} \:,
    \end{equation*}
i.e.\ the map~$\Phi$ is a~$^{\ast}$-algebra homomorphism. Recall the discussion at the beginning of this section on the symmetry of the maps~$\gamma_{\eta}^{j}$ with respect to the spin space inner product~$\Sl\cdot|\cdot\Sr_{S_{x}M}$. The adjoint operator~$g(H_{\eta})^{\ast}$ is
    \begin{align*}
        (\psi|g(H_{\eta})\varphi)_{t}&=\int_{\mathbb{R}^{3}} \Sl\psi(x)|\gamma_{\eta0}(g(H_{\eta})\varphi)(x)\Sr_{S_{x}M} d^{3}x\\
        &=2\pi \int_{\mathbb{R}^{3}}\int_{\mathbb{R}^{3}} \Sl\gamma_{\eta0}\psi(x)|Q_{g}(x,y)\gamma_{\eta0}\varphi(y)\Sr_{S_{x}M} d^{3}y \;d^{3}x\\
        &=2\pi \int_{\mathbb{R}^{3}}\int_{\mathbb{R}^{3}} \Sl Q_{\overline{g}}(y,x)\gamma_{\eta}^{0}\psi(x)|\gamma_{\eta0}\varphi(y)\Sr_{S_{y}M} d^{3}y \;d^{3}x \\
        &=\int_{\mathbb{R}^{3}} \Sl (\overline{g}(H_{\eta})\psi)(y)|\gamma_{\eta0}\varphi(y)\Sr_{S_{y}M} d^{3}y =(\overline{g}(H_{\eta})\psi|\varphi)_{t}=(g(H_{\eta})^{\ast}\psi|\varphi)_{t}\:,
    \end{align*}
    where~$Q_{\overline{g}}(y,x)$ is the adjoint of~$Q_{g}(x,y)$ with respect to~$\Sl\cdot|\cdot\Sr_{S_{x}M}$.
\end{itemize}
Thus we obtain the functional calculus associated to the self-adjoint operator~$H_{\eta}$. Considering now the interval~$I=(-\frac{1}{\varepsilon},\omega)$, the characteristic function~$g=\chi_{I}$ and using that~$\chi_{I}(k^{0})=\Theta(1+\varepsilon k^{0})\Theta(-k^{0}+\omega)$ we recover expression~\eqref{chiIH} with~$Q_{g}=P^{\varepsilon,\omega}(x,y)$ as in equation~\eqref{Pepsdef_omega}. The~$(-1)$ factor in~\eqref{chiIH} is simply due to the sign function~$\epsilon(k^{0})$.

That~$F_{\omega}(H_{\eta})$ is an integral operator follows directly from the fact that~$\chi_{I}(H_{\eta})$ is an integral operator and that~$F_{\omega}(H_{\eta})=\frac{d}{d\omega}\big(\chi_{(-\frac{1}{\varepsilon},\omega)}(H_{\eta}))$
    \begin{align*}
        (F_{\omega}(H_{\eta})\psi)(x)&=\Big(\Big(\frac{d}{d\omega}\chi_{I}(H_{\eta})\Big)\psi\Big)(x)=-2\pi\int_{\mathbb{R}^{3}}\Big(\frac{d}{d\omega}P^{\varepsilon,\omega}(x,y)\Big)\gamma_{\eta0}\psi(y)d^{3}y\\
        &=:\int_{\mathbb{R}^{3}}F_{\omega}(x,y)\psi(y)d^{3}y\:,
    \end{align*}
i.e.\ $F_{\omega}(x,y)=-2\pi\frac{d}{d\omega}P^{\varepsilon,\omega}(x,y)\gamma_{\eta0}$. The proof that the integral kernel~$F_{\omega}(x,y)$ satisfies expressions~\eqref{eq:F_omega} and~\eqref{eq:partial_F_omega} can be found in the Appendix.
\end{proof}

\begin{Lemma}\label{lem:trace_class}
Let~$(N_{t})_{t\in\mathbb{R}}$ be the foliation of Minkowski spacetime given by the level sets of the
    global time function~$t$. Given an initial time~$t_0$, we consider a compact subset~$V\subset N_{t_{0}}$, a positive
    and smooth function~$\Tilde{f}\in C^{\infty}(\mathbb{R}^{3},\mathbb{R}_{>0})$ on the initial Cauchy surface~$N_{t_{0}}$ (with~$\Tilde{f}_{p}=\Tilde{f} (x,y,z)$) and a spacelike vector field~$X$. Furthermore, assume that~$\Tilde{f}_{p}=1$ for all~$p\in N_{t_{0}}\setminus V$ and that the vector field~$X$ vanishes outside the compact subset~$V$. Consider the regularizing vector field~$u : M \rightarrow TM$ constructed from the locally rigid evolution (cf. Lemma~\ref{lem:dyn_eq}) of the initial vector field 
\begin{equation*}
    u_{p}=(1+\lambda\Tilde{f}_{p})\nu+\lambda X_{p} \qquad \text{for all~$p\in N_{t_{0}}$}
\end{equation*}
    (where~$\nu$ is again the normal vector field to~$N_{t_{0}}$ and~$\lambda\geq0$ is an arbitrary constant). 
    Then, $\dot{A}_{t}F_{\omega'}(H_\eta)$, $Q(\omega',\omega'')\coloneqq\Delta A(t)F_{\omega'}(H_{\eta})\Delta A(t)F_{\omega''}(H_{\eta})$ and~$\frac{d}{dt}Q(\omega',\omega'')$ (with~$\Delta A(t)\coloneqq A_{t}-H_{\eta}$) are trace-class operators for all~$\omega',\omega''\in\sigma(H_{\eta})$(where~$\dot{A}_{t}\coloneqq\frac{dA_{t}}{dt}$ is expanded only to linear order in~$\lambda$ and~$H_{\eta}$ denotes the Dirac Hamiltonian, cf. expression~\eqref{eq:Dirac_Hamiltonian}).
\end{Lemma}
\begin{proof}
In the first place, note that for~$p\in N_{t_{0}}\setminus V$, $u_{p}=\partial_{t}$ and so ~$\frac{du_{p}}{dt}=\frac{dA_{t}}{dt}=0$. For~$p\in V$, using the dynamical equation~\eqref{eq:dynamical_eq_u_Minkowski}, the series expansion to first order in~$\lambda$ of~$\frac{dA_{t}}{dt}$ yields
      \begin{align}
        A_{t} &= \frac{1}{2}\:\big\{u^{t},-i\gamma_{\eta t}\gamma_{\eta}^{\alpha}\partial_{\alpha}+m\gamma_{\eta0}\big\}
        +\frac{i}{2} \:\big\{u^{\alpha},\partial_{\alpha} \big\}\nonumber\\
        \dot{A}_{t}&\coloneqq \frac{dA_{t}}{dt} =\frac{1}{2} \:\Big\{\frac{du^{t}}{dt},-i\gamma_{\eta t}\gamma_{\eta}^{\alpha}\partial_{\alpha}+m\gamma_{\eta0}\Big\}+\frac{i}{2} \:\Big\{\frac{du^{\alpha}}{dt},\partial_{\alpha} \Big\}\nonumber\\
        &=\frac{\lambda}{2}\: \Big\{4\frac{(\partial_{\mu}f)}{f^{3}}h+\frac{(\partial_{\mu} h)}{3f^{2}},H_{\eta} \Big\}
        -\frac{i}{2}\:\big\{\delta^{\alpha\beta} (\partial_{\alpha}f^{-1}),\partial_{\beta} \big\}+\mathcal{O}(\lambda^{2})\nonumber\\
        &=\frac{\lambda}{2}\: \bigg[\left(8\frac{(\partial_{\mu} f)}{f^{3}}h +\frac{2\,(\partial_{\mu}h)}{3f^{2}}\right)H_{\eta}+ \bigg( H_{\eta} \Big( 4\frac{(\partial_{\mu} f)}{f^{3}}h+\frac{(\partial_{\mu} h)}{3f^{2}}\Big) \bigg) \bigg] \nonumber\\
        &-\frac{i}{2}\left(2\delta^{\alpha\beta}(\partial_{\alpha}f^{-1})\:\partial_{\beta}+ \big(\Delta_{\delta} f^{-1} \big)\right)+\mathcal{O}(\lambda^{2}) \label{eq:dotA}\:,
    \end{align}
where~$\Delta_{\delta}$ is the Laplacian operator (in Euclidean space~$(\mathbb{R}^{3},\delta)$) and ~$f_{p}\coloneqq 1+\lambda\Tilde{f}_{p}$. Here, as discussed at the beginning of the proof of Lemma~\ref{lem:dyn_eq}, we assumed again without loss of generality that, in local coordinates, the initial regularizing vector field is given by
\[ u_{p}=f_{p}\partial_{t}+\lambda h_{p}\partial_{\mu} \qquad \text{for all~$p\in N_{t_{0}}$}\:, \] with~$\mu\in\{1,2,3\}$, where~$h\in C^{\infty}(\mathbb{R}^{3},\mathbb{R}_{>0})$ is a function of all three~$(x,y,z)$ coordinates.

Clearly, the coefficients of the differential operator~$\dot{A}_{t}$ are smooth and, by the previous discussion, compactly supported in~$N_{t}$ as~$u=\partial_{t}$ outside a compact subset~$V\subset N_{t}$. On the other hand, as proven in Lemma~\ref{lem:spectral_measure}, ~$F_{\omega'}(H_\eta)$ is an integral operator with smooth kernel. Therefore, the operator product~$\dot{A}_{t}F_{\omega'}(H_\eta)$ is an integral operator of the form
\[ \bigg( \Big(\frac{d}{dt}A_{t} \Big) F_{\omega'}(H_{\eta}) \psi \bigg)(x) = (K \psi)(x) := \int_{\R^3} K(x,y)\:
\psi(y)\: d^3y \]
with~$x=(t,x^{\mu}), y=(0,y^{\mu})$, $\psi\in C^{\infty}_{\textrm{sc}}(M,SM)\cap\H_{m}$ and a smooth, compactly supported (in~$V$) kernel~$K(x,y)$.
Now we can follow the procedure in~\cite[Section~30, Theorem~13]{lax}:
The operator~$K^* K$ is positive and has again a smooth, compactly supported kernel.
It then follows from Mercer's theorem that this operator is trace class and thus that it has a purely discrete spectrum.
We order its eigenvalues in descending order~$\lambda_1 \geq \lambda_2 \geq \cdots \geq 0$.
Moreover, approximating the kernel by polynomials (using the Stone-Weierstra{\ss} theorem),
one concludes as in~\cite[Section~30, Theorem~13]{lax} that the eigenvalues~$\lambda_n$ of~$K^* K$
decay faster than any polynomial (i.e.\ for any~$b>0$ there is a
constant~$c>0$ such that~$\lambda_n \leq c\, n^{-b/2}$ for all~$n \in \N$). 
A decay faster than quadratic (i.e.\ $b>4$) implies that~$K$ is trace class (and thus also the product operator~$\dot{A}_{t}F_{\omega'}(H_{\eta})$) as the decay of the eigenvalues of~$K^{\ast}K$ guarantees that the series given by the sum over all singular values of~$K$ (i.e.\ the non-zero eigenvalues of~$\sqrt{K^{\ast}K}$) converges:
\begin{equation*}
    \sum_{j=0}^{\infty} s_{n}(K)\leq C \sum_{n=0}^{\infty}\frac{1}{n^{\beta}}<\infty
\end{equation*}
where~$\beta:=b/4>1$, $C$ is a constant and~$s_{n}(K)$ denotes the $n^\text{th}$ singular value of~$K$ (note that from the previous bound on~$\lambda_{n}$, we have that~$s_{n}(K)=\sqrt{\lambda_{n}}\leq \sqrt{c}n^{-b/4}$).
\par By the same arguments it follows that the product operator~$Q_{\omega',\omega''}$ and~$\frac{d}{dt}Q(\omega',\omega'')$, with~$\omega',\omega''\in\sigma(H_{\eta})$ and~$\Delta A(t)\coloneqq A_{t}-H_{\eta}$, are trace-class: outside the compact set~$V$ the operator~$A_{t}$ agrees with~$H_{\eta}$, so~$\Delta A(t)$ is a differential operator with smooth compactly supported coefficients (and the same holds for~$\frac{d}{dt}\Delta A(t)=\dot{A}_{t}$). Moreover, $F_{\omega'}(H_{\eta})$ is an integral operator, so the argument given above to show that
an operator is trace-class also applies to~$Q(\omega',\omega'')$ and~$\frac{d}{dt}Q(\omega',\omega'')$.
\end{proof}

The following lemma gives us a distributional equation which will be helpful in order to determine the second order term of the rate of baryogenesis.

\begin{Lemma}\label{lem:sokhotski}
    Consider an arbitrary smooth and compactly supported function~$f\in C^{\infty}_{0}(\mathbb{R}\times\mathbb{R},\mathbb{C})$. Then, the following identity holds in the distributional sense.
    \begin{align*}
    &\slim_{\delta\to0^{+}}\int_{-\infty}^{\infty}d\omega'\int_{-\infty}^{\infty}d\omega''f(\omega',\omega'')\Big[\frac{\partial}{\partial\omega'}\Big(\frac{1}{\omega'-\omega-is\delta}\Big)\Big(\frac{1}{\omega''-\omega-is\delta}\Big)\Big]\Big|^{s=1}_{s=-1}\nonumber\\
    &=2\pi i \int_{-\infty}^{\infty} d\omega'\: \Big( \partial_{\omega}f(\omega,\omega') + 
\partial_{\omega'}f(\omega',\omega) \Big)\: \frac{1}{\omega'-\omega}
    \end{align*}
    where~$\omega\in\mathbb{R}$, $\partial_{\omega}\coloneqq\frac{\partial}{\partial\omega}$ and~$\partial_{\omega'}\coloneqq\frac{\partial}{\partial\omega'}$.
\end{Lemma} \noindent
For clarity, we explain right away how the right side of our formula defines a distribution.
Thus let~$g \in C^\infty_0(\R)$ be a test function. Then, using the symmetry properties of the
integrand, we have
\begin{align*}
&\int_{-\infty}^\infty d\omega\: g(\omega) \int_{-\infty}^{\infty} d\omega'\: \Big( \partial_{\omega}f(\omega,\omega') + 
\partial_{\omega'}f(\omega',\omega) \Big)\: \frac{1}{\omega'-\omega} \\
&= -\frac{1}{2} \int_{-\infty}^\infty d\omega \int_{-\infty}^\infty d\omega' \Big( \partial_{\omega}f(\omega,\omega') + 
\partial_{\omega'}f(\omega',\omega) \Big)\: \frac{g(\omega') - g(\omega)}{\omega'-\omega} \:,
\end{align*}
and the last line is clearly well-defined.
\begin{proof}
    In the first place, we rewrite a polynomial similar to the one appearing in the integrand using partial fraction decomposition
    \begin{equation}\label{eq:fraction_decomp}
        \Big(\frac{1}{\omega'-\omega-is\delta}\Big)\Big(\frac{1}{\omega''-\omega-is\delta}\Big)=\frac{1}{\omega''-\omega'}\Big(\frac{1}{\omega'-\omega-is\delta}-\frac{1}{\omega''-\omega-is\delta}\Big)\:,
    \end{equation}
    where, for now, we assume that~$\omega'\neq\omega''$. On the other hand, by the Sokhotski-Plemelj formula it holds that
    \begin{align}
        &\slim_{\delta\to0^{+}}\Big(\frac{1}{\omega'-\omega-is\delta}\Big)=\frac{1}{2}\frac{\textrm{PP}}{\omega'-\omega}+i\pi s \delta(\omega'-\omega)\:,\label{eq:sokhotski1}
    \end{align}
    where~$s\in\{-1,1\}$ and~$\textrm{PP}$ denotes the Cauchy principal value, which is a tempered distribution defined by
    \begin{equation*}
        \frac{\textrm{PP}}{\omega'-\omega}(f)=\lim_{\varepsilon\to0}\int_{\mathbb{R}\setminus(\omega-\varepsilon,\omega+\varepsilon)}\frac{f(\omega')}{\omega'-\omega}d\omega' \;,
    \end{equation*}
    where~$f\in C_{0}^{\infty}(\mathbb{R})$. Hence, by taking the limit~$\delta\to0^{+}$ in equation~\eqref{eq:fraction_decomp} and the difference between the~$s=1$ and~$s=-1$ terms the following distributional equation is obtained:
    \begin{align}
        &\slim_{\delta\to0^{+}}\Big[\Big(\frac{1}{\omega'-\omega-is\delta}\Big)\Big(\frac{1}{\omega''-\omega-is\delta}\Big)\Big]\Big|^{s=1}_{s=-1}\nonumber\\
        &=\frac{1}{\omega''-\omega'}\Big[\frac{1}{2}\frac{\textrm{PP}}{\omega'-\omega}+i\pi s \delta(\omega'-\omega)-\frac{1}{2}\frac{\textrm{PP}}{\omega''-\omega}-i\pi s \delta(\omega''-\omega)\Big]\Big|_{s=-1}^{s=1}\nonumber\\
        &=2\pi i\frac{1}{\omega''-\omega'}\Big(\delta(\omega'-\omega)-\delta(\omega''-\omega)\Big)\:.\label{eq:limit_fraction}
    \end{align}
    Note that from the previous expression we can also make sense out of the case~$\omega'=\omega''$: if we take the limit~$\omega''\to\omega'$ the right hand side of the previous equation is simply~$-2\pi i \frac{\partial}{\partial\omega'}\delta(\omega'-\omega)$, where~$\frac{\partial}{\partial\omega'}$ is to be understood as a distributional derivative.  In order to obtain the polynomial appearing in the integral of which we want to take the limit~$\delta\to0^{+}$, it only remains to take the distributional derivative~$\frac{\partial}{\partial\omega'}$ of~\eqref{eq:limit_fraction} and apply the chain rule. We consider now a function~$f\in C^{\infty}_{0}(\mathbb{R}\times\mathbb{R},\mathbb{C})$ and integrate over~$\omega'$ and~$\omega''$ as in the statement of the lemma. Integrating by parts enables us to use expression~\eqref{eq:limit_fraction} for the limit~$\delta\to0^{+}$
    \begin{align*}
&\slim_{\delta\to0^{+}}\int_{-\infty}^{\infty}d\omega'\int_{-\infty}^{\infty}d\omega''\:
f(\omega',\omega'') \Big[\frac{\partial}{\partial\omega'}\Big(\frac{1}{\omega'-\omega-is\delta}\Big)\Big(\frac{1}{\omega''-\omega-is\delta}\Big)\Big]\Big|^{s=1}_{s=-1} \\
&= -\slim_{\delta\to0^{+}}\int_{-\infty}^{\infty}d\omega'\int_{-\infty}^{\infty} d\omega''\:
\partial_{\omega'}f(\omega',\omega'') \:\Big[\Big(\frac{1}{\omega'-\omega-is\delta}\Big)\Big(\frac{1}{\omega''-\omega-is\delta}\Big)\Big]\Big|^{s=1}_{s=-1} \\
&=-2 \pi i
\int_{-\infty}^{\infty}d\omega'\int_{-\infty}^{\infty} d\omega''\:
\partial_{\omega'}f(\omega',\omega'') \:
\frac{1}{\omega''-\omega'}\Big(\delta(\omega'-\omega)-\delta(\omega''-\omega)\Big) \\
&=-2 \pi i \int_{-\infty}^{\infty} \partial_{\omega}f(\omega,\omega'') \: \frac{1}{\omega''-\omega}\: d\omega''
+ 2 \pi i \int_{-\infty}^{\infty} \partial_{\omega'}f(\omega',\omega) \: \frac{1}{\omega-\omega'}\:d\omega' \\
&=-2 \pi i \int_{-\infty}^{\infty} \partial_{\omega}f(\omega,\omega') \: \frac{1}{\omega'-\omega}\: d\omega'
+ 2 \pi i \int_{-\infty}^{\infty} \partial_{\omega'}f(\omega',\omega) \: \frac{1}{\omega-\omega'}\:d\omega' \\
&=-2 \pi i \int_{-\infty}^{\infty} \Big( \partial_{\omega}f(\omega,\omega') + 
\partial_{\omega'}f(\omega',\omega) \Big)
\: \frac{1}{\omega'-\omega}\: d\omega'\:,
\end{align*}
where in the fifth line we simply relabeled the integration variable of the first integral (i.e.\ $\omega''$ was replaced by~$\omega'$) in order to assemble the obtained result as one integral.
\end{proof}

Using the previous lemmata we can compute the rate of baryogenesis to first and second order in~$\lambda$, which is the main result of this section. 

\begin{Thm}\label{theo:baryo_rate_Minkowski}
Consider the setup of Lemma \ref{lem:trace_class}. Furthermore, assume that the spectrum of~$A_{t}$ is absolutely continuous. Then, the rate of baryogenesis~$B_{t}$ to first and second order in~$\lambda$ are given by
    \begin{align}
    &B^{(1)}_{t}=\tr_{\H_{t}^{\varepsilon}}\left(\dot{A}_{t}F_{-m}(H_\eta)\right) \label{B1} \\
    &B^{(2)}_{t}=-\int_{-\infty}^\infty d\omega \int_{-\infty}^\infty d\omega'\:
    \partial_{\omega}\big(\eta_{\Lambda}(\omega)\tr_{\H_{t}^{\varepsilon}}(\dot{Q}(\omega,\omega'))\big)\:\frac{g(\omega') - g(\omega)}{\omega'-\omega}\:, \label{B2}
    \end{align}
    where~$Q(\omega,\omega')\coloneqq F_{\omega}\Delta AF_{\omega'}\Delta A$, $\dot{Q}(\omega,\omega')\coloneqq\frac{d}{dt}Q(\omega,\omega')$ with~$\Delta A=A_{t}-H_{\eta}$, and~$g$ is the characteristic function of the set~$(-1/\varepsilon,-m)$.
Moreover, $B_{t}^{(1)}$ vanishes.
\end{Thm} \noindent
This theorem shows in particular that, if the regularizing vector field~$u=\partial_t$ is constant and
normal to the Cauchy surface, then~$\Delta A=0$, and no baryogenesis occurs.
 It is remarkable that the first order contribution always vanishes; this can be understood
from the fact that the kernel~$F_{-m}(H_\eta)$ and its first spatial derivatives vanish on the diagonal.
The second order contribution, however, is in general non-zero, as the following consideration shows:
Suppose we choose~$\Delta A$ such that it maps the image of~$F_{\omega}(H_\eta)$ to the image of~$F_{\omega'}(H_\eta)$
with~$\omega<-m$ and~$\omega'>m$. Then the trace of~$Q(\omega, \omega')$ is strictly positive,
and the difference quotient in~\eqref{B2} is non-zero.

\begin{proof}[Proof of Theorem~\ref{theo:baryo_rate_Minkowski}.]
    In the first place, we start by discussing how to obtain an expansion of the rate of baryogenesis in powers of~$\lambda$, following a strategy very close to the one used to derive equation (7.3) in~\cite[Section 7.3]{baryogenesis}.

Let~$A_{t}=H_{\eta}+\Delta A(t)$ and for~$\omega$ in the resolvent set of~$H_\eta$ and~$A_{t}$
consider the resolvents $R_{\omega}(H_{\eta})=(H_{\eta}-\omega)^{-1}$ and $R_{\omega}(A_{t})=(A_{t}-\omega)^{-1}$. Then, $R_{\omega}(A_{t})$ and $R_{\omega}(H_{\eta})$ are related as follows (cf. \cite[equation (7.2)]{baryogenesis}):
    \begin{equation}\label{eq:expansion_resolvent}
        R_{\omega}(A_{t})=\sum_{p=0}^{\infty}(-R_{\omega}(H_{\eta})\Delta A(t))^{p}R_{\omega}(H_{\eta})=:\sum_{p=0}^{\infty}R_{\omega}^{(p)}(A_{t})\:.
    \end{equation}
     In particular, the first and second order terms of the previous expansion are
     \begin{align}
         &R_{\omega}^{(1)}(A_{t})=-R_{\omega}(H_{\eta})\Delta A(t)R_{\omega}(H_{\eta})\\
         &R_{\omega}^{(2)}(A_{t})= R_{\omega}(H_{\eta})\Delta A(t)R_{\omega}(H_{\eta})\Delta A(t)R_{\omega}(H_{\eta})\:.\label{eq:second_order_resolvent}
     \end{align}   
     Assuming that~$A_{t}$ has an absolutely continuous spectrum and applying Stone's formula (see~\cite[Theorem VII.13]{reed+simon}) to the projection valued measure~$\chi_{I}(A_{t})$ with~$I\coloneqq(-1/\varepsilon,-m)$ we obtain that
    \begin{equation*}
       \frac{1}{2}\big(\chi_{I}(A_{t})+\chi_{\overline{I}}(A_{t})\big)=\chi_{I}(A_{t})=\frac{1}{2\pi i}\slim_{\delta\to0^{+}}\int_{-\frac{1}{\varepsilon}}^{-m}(R_{\omega+i\delta}(A_{t})-R_{\omega-i\delta}(A_{t}))d\omega\;.
    \end{equation*}
    Hence, from this expression and writing the projection valued measure associated to~$A_{t}$ as~$\chi_{I}(A_{t})=\int_{-1/\varepsilon}^{-m}F_{\omega}(A_{t})d\omega$ (similar to what we did in Lemma~\ref{lem:spectral_measure}),  we can directly read off~$F_{\omega}(A_{t})$,
    \begin{align}
        &F_{\omega}(A_{t})=\frac{1}{2\pi i}\slim_{\delta\to0^{+}}\big(R_{\omega+i\delta}(A_{t})-R_{\omega-i\delta}(A_{t})\big)\;.
    \end{align}
        Using the spectral calculus of the selfadjoint operators~$A_{t}$ and~$H_{\eta}$, we can express the product operator~$\eta_{\Lambda}(H_{\eta})\chi_{I}(A_{t})$ by:
    \begin{align}
        \eta_{\Lambda}(H_{\eta})\chi_{I}(A_{t})&=\int_{-\frac{1}{\varepsilon}}^{-m}\eta_{\Lambda}(H_{\eta})F_{\omega}(A_{t})d\omega=\int_{-\frac{1}{\varepsilon}}^{-m}d\omega\int_{-\infty}^{\infty}d\omega'\eta_{\Lambda}(\omega')F_{\omega'}(H_{\eta})F_{\omega}(A_{t})\nonumber\\
        &=\frac{1}{2\pi i}\slim_{\delta\to0^{+}}\int_{-\frac{1}{\varepsilon}}^{-m}d\omega\int_{-\infty}^{\infty}d\omega'\eta_{\Lambda}(\omega')F_{\omega'}(H_{\eta})R_{\omega+is\delta}(A_{t})\big|_{s=-1}^{s=1}\label{eq:product_op}
        \end{align}
        In particular using the power expansion~\eqref{eq:expansion_resolvent} for the resolvent we can also express the rate of baryogenesis (recall Definition~\ref{def:Baryo}) as a power expansion
        \begin{equation}
            B_{t}=\sum_{p=0}^{\infty}B^{(p)}_{t}\;,
        \end{equation}
        where the~$p^\text{th}$ term of this expansion has the form
        \begin{equation}\label{eq:Bt_expansion}
            B_{t}^{(p)}\coloneqq \frac{d}{dt}\Big(\frac{1}{2\pi i}\slim_{\delta\to0^{+}}\int_{-\frac{1}{\varepsilon}}^{-m}\!\!\!d\omega\int_{-\infty}^{\infty} \!\!\!d\omega'\eta_{\Lambda}(\omega')\tr_{\H_{t}^{\varepsilon}}\big(F_{\omega'}(H_{\eta})R^{(p)}_{\omega+is\delta}(A_{t})\big)\big|_{s=-1}^{s=1}\Big)\:,
        \end{equation}
        and, importantly, carries a power~$\lambda^{p}$ (which stems from the operator~$\Delta A(t)=A_{t}-H_{\eta}$ appearing in expression~\eqref{eq:expansion_resolvent}, where $\dot{A}_{t}$ is expanded to linear order in~$\lambda$ as in equation \eqref{eq:dotA}). On the other hand, in Minkowski spacetime it holds for all times~$t\in\mathbb{R}$ that~$\H^{\varepsilon}_{t}\subset L^{2}(\mathbb{R}^{3},\mathbb{C}^{4})$, which is a time-independent Hilbert space. For this reason, $\frac{d}{dt}\tr_{\H^{\varepsilon}_{t}}(\cdot)$ is well-defined and satisfies that~$\frac{d}{dt}\tr_{\H^{\varepsilon}_{t}}(Q)=\tr_{\H^{\varepsilon}_{t}}\big(\frac{d}{dt}Q\big)$ for an arbitrary trace-class operator~$Q : \H^{\varepsilon}_{t}\to\H^{\varepsilon}_{t}$.
        \par Let us now derive a formula for the rate of baryogenesis to linear order in $\lambda$. We start by rewriting the integral with respect to $\omega'$ as follows:
        \begin{align*}
            &\int_{-\infty}^{\infty}d\omega'\eta_{\Lambda}(\omega')\tr_{\H_{t}^{\varepsilon}}\big(F_{\omega'}(H_{\eta})R^{(1)}_{\omega+is\delta}(A_{t})\big)\big|_{s=-1}^{s=1}\\
            &=-\int_{-\infty}^{\infty}d\omega'\eta_{\Lambda}(\omega')\tr_{\H_{t}^{\varepsilon}}\big(F_{\omega'}(H_{\eta})R_{\omega+is\delta}(H_{\eta})\Delta A(t)R_{\omega+is\delta}(H_{\eta})\big)\big|_{s=-1}^{s=1}\\
            &=-\int_{-\infty}^{\infty}d\omega'\eta_{\Lambda}(\omega')\tr_{\H_{t}^{\varepsilon}}\big(F_{\omega'}(H_{\eta})(R_{\omega+is\delta}(H_{\eta}))^{2}\Delta A(t)\big)\big|_{s=-1}^{s=1}\\
            &=-\int_{-\infty}^{\infty}d\omega'\eta_{\Lambda}(\omega')\frac{1}{(\omega'-\omega-is\delta)^{2}}\Big|_{s=-1}^{s=1}\tr_{\H_{t}^{\varepsilon}}\big(F_{\omega'}(H_{\eta})\Delta A(t)\big)\\
            &=\int_{-\infty}^{\infty}d\omega'\eta_{\Lambda}(\omega')\frac{\partial}{\partial\omega'}\Big(\frac{1}{\omega'-\omega-is\delta}\Big)\tr_{\H_{t}^{\varepsilon}}\big(F_{\omega'}(H_{\eta})\Delta A(t)\big)\:.
        \end{align*}
    We can now use this to determine the linear term in the power expansion of the rate of baryogenesis,
    \begin{align}
        B^{(1)}_{t}&=\frac{d}{dt} \Big(\frac{1}{2\pi i}\slim_{\delta\to0^{+}}\int_{-\frac{1}{\varepsilon}}^{-m}d\omega\int_{-\infty}^{\infty}d\omega'\eta_{\Lambda}(\omega')\tr_{\H_{t}^{\varepsilon}}\big(F_{\omega'}(H_{\eta})R^{(1)}_{\omega+is\delta}(A_{t})\big)\big|_{s=-1}^{s=1}\Big)\nonumber\\
        &=\frac{1}{2\pi i}\slim_{\delta\to0^{+}}\int_{-\frac{1}{\varepsilon}}^{-m} \!\!\! d\omega\int_{-\infty}^{\infty}\!\!\!d\omega'\eta_{\Lambda}(\omega') \:\tr_{\H_{t}^{\varepsilon}}\big(F_{\omega'}(H_{\eta})\dot{A}_{t}\big)\frac{\partial}{\partial\omega'}\Big(\Big(\frac{1}{\omega'-\omega-is\delta}\Big)\Big|_{s=-1}^{s=1}\Big)\nonumber \\
        &=\int_{-\frac{1}{\varepsilon}}^{-m}d\omega\int_{-\infty}^{\infty}d\omega'\eta_{\Lambda}(\omega')\tr_{\H_{t}^{\varepsilon}}\big(F_{\omega'}(H_{\eta})\dot{A}_{t}\big)\frac{\partial}{\partial\omega'}\delta(\omega'-\omega)\nonumber\\
        &=-\int_{-\frac{1}{\varepsilon}}^{-m}d\omega\frac{\partial}{\partial\omega}\bigg(\eta_{\Lambda}(\omega)\tr_{\H_{t}^{\varepsilon}}\big(F_{\omega}(H_{\eta})\dot{A}_{t}\big)\bigg)\nonumber\\
        &=-\eta_{\Lambda}(-m)\tr_{\H_{t}^{\varepsilon}}\big(F_{-m}(H_{\eta})\dot{A}_{t}\big)+\eta_{\Lambda}(-1/\varepsilon)\tr_{\H_{t}^{\varepsilon}}\big(F_{-1/\varepsilon}(H_{\eta})\dot{A}_{t}\big)\nonumber\\
        &=-\tr_{\H_{t}^{\varepsilon}}\big(F_{-m}(H_{\eta})\dot{A}_{t}\big)\:,\label{eq:linear_rate_baryo}
    \end{align}
    where in the second line we used that~$\frac{d}{dt}\Delta A(t)=\dot{A}_{t}$, in the fourth line we integrated by parts and in the last line we used that~$\eta_{\Lambda}(-\frac{1}{\varepsilon})=0$ (recall that~$m\ll \Lambda\ll\frac{1}{\varepsilon}$ and~$\eta_{\Lambda}$ is compactly supported in~$(-\Lambda,\Lambda)$) and normalized the cutoff function~$\eta_{\Lambda}$ such that~$\eta_{\Lambda}(-m)=1$. Note that the product operator~$\dot{A}_{t}F_{-m}(H_{\eta})$ is trace-class by Lemma~\ref{lem:trace_class}. 
    \par We now proceed to compute the ratio of baryogenesis~$B_t$ to first order in~$\lambda$ starting from the expression~\eqref{eq:linear_rate_baryo} we just derived (with~$\dot{A}_{t}$ linear in~$\lambda$ and given by equation~\eqref{eq:dotA}). It is known that the trace of an integral operator with an integral kernel in~$C^{0}([0,1])$ and which is trace class is given by the integral over the diagonal elements of its kernel (cf.~\cite[Section~30, Theorem~13]{lax}; for a generalization of this statement to integral operators in~$L^{2}(\mathbb{R}^{n})$, see~\cite[Theorem 3.1]{brislawn}). We use this
in order to compute the trace of the integral operator~$\dot{A}_{t}F_{-m}(A_{t})$, which has a compactly supported (in~$V$) integral kernel. Hence, combining expressions~\eqref{eq:linear_rate_baryo} and Lemma~\ref{lem:spectral_measure} it follows that
    \begin{align}\label{eq:rewrite_rate_baryo}
        B_{t}^{(1)}&
        =-\int_{V}\Tr_{\C^4}\big(\dot{A}_{t}F_{-m}(x,y)\big)\Big|_{y=x}d^{3}x\:,
    \end{align}
where we restricted the region of integration from all of~$\mathbb{R}^{3}$ to~$V$ as~$\dot{A}_{t}$ vanishes outside of this region. Note that inserting equation ~\eqref{eq:dotA} for~$\dot{A}_{t}$ in the previous expression yields integrals (linear in~$\lambda$) that involve~$F_{-m}(x,x)$ and~$\left(\partial_{\mu}F_{-m}(x,y)\right)|_{-m}$, which are identically zero (expressions~\eqref{eq:F_omega} and~\eqref{eq:partial_F_omega} vanish for~$\omega=-m$; their derivation can be found in the Appendix). Hence, the rate of baryogenesis vanishes to first order in~$\lambda$.
\par Finally, let us determine the rate of baryogenesis to second order in~$\lambda$. Using expression~\eqref{eq:Bt_expansion} and rewriting the integral over~$\omega$ with the help of the characteristic function~$g$ of the set~$(-1/\varepsilon,-m)$, we have that
    \begin{align}
        & B^{(2)}_{t} =\frac{d}{dt}\Big(\frac{1}{2\pi i}\slim_{\delta\to0^{+}}\int_{-\frac{1}{\varepsilon}}^{-m}d\omega\int_{-\infty}^{\infty}d\omega'\eta_{\Lambda}(\omega')\tr_{\H_{t}^{\varepsilon}}\big(F_{\omega'}(H_{\eta})R^{(2)}_{\omega+is\delta}(A_{t})\big)\big|_{s=-1}^{s=1}\Big)\nonumber\\
        & =\frac{d}{dt}\Big(\frac{1}{2\pi i}\slim_{\delta\to0^{+}}\int_{-\infty}^{\infty}g(\omega)d\omega\int_{-\infty}^{\infty}d\omega'\eta_{\Lambda}(\omega')\tr_{\H_{t}^{\varepsilon}}\big(F_{\omega'}(H_{\eta})R^{(2)}_{\omega+is\delta}(A_{t})\big)\big|_{s=-1}^{s=1}\Big)\:. \label{eq:Bt2}
    \end{align}
    The advantage of this expression is that it will allow us to use afterward Lemma~\ref{lem:sokhotski} in order to take the limit~$\delta\to0^{+}$. We now rewrite the integrand using expression~\eqref{eq:second_order_resolvent}
\begin{align*}
    &\eta_{\Lambda}(\omega')\tr_{\H_{t}^{\varepsilon}}\big(F_{\omega'}(H_{\eta})R^{(2)}_{\omega+is\delta}(A_{t})\big)\big|_{s=-1}^{s=1}\\
    &=\eta_{\Lambda}(\omega')\tr_{\H_{t}^{\varepsilon}}\big(F_{\omega'}(H_{\eta})(R_{\omega+is\delta}(H_{\eta}))^{2}\Delta A(t)R_{\omega+is\delta}(H_{\eta})\Delta A(t)\big)\big|_{s=-1}^{s=1}\\
    &=\int_{-\infty}^{\infty}d\omega''\eta_{\Lambda}(\omega')\Big[\Big(\frac{1}{(\omega'-\omega-is\delta)^{2}}\Big)\Big(\frac{1}{\omega''-\omega-is\delta}\Big)\Big]\Big|^{s=1}_{s=-1} \\
    &\qquad\qquad\qquad\qquad\qquad\qquad\qquad\qquad \times \tr_{\H_{t}^{\varepsilon}}\big(F_{\omega'}(H_{\eta})\Delta AF_{\omega''}(H_{\eta})\Delta A\big)\\
    &=-\int_{-\infty}^{\infty}d\omega''\eta_{\Lambda}(\omega')\Big[\frac{\partial}{\partial\omega'}\Big(\frac{1}{\omega'-\omega-is\delta}\Big)\Big(\frac{1}{\omega''-\omega-is\delta}\Big)\Big]\Big|^{s=1}_{s=-1} \\
    &\qquad\qquad\qquad\qquad\qquad\qquad\qquad\qquad \times \tr_{\H_{t}^{\varepsilon}}\big(F_{\omega'}(H_{\eta})\Delta AF_{\omega''}(H_{\eta})\Delta A\big)\\
    &=-\int_{-\infty}^{\infty}d\omega''\eta_{\Lambda}(\omega')\tr_{\H_{t}^{\varepsilon}}(Q(\omega',\omega''))\Big[\frac{\partial}{\partial\omega'}\Big(\frac{1}{\omega'-\omega-is\delta}\Big)\Big(\frac{1}{\omega''-\omega-is\delta}\Big)\Big]\Big|^{s=1}_{s=-1}\:,
\end{align*}
where~$Q(\omega',\omega'')\coloneqq F_{\omega'}(H_{\eta})\Delta AF_{\omega''}(H_{\eta})\Delta A$ is the product operator appearing in the trace. Note that both~$Q(\omega',\omega'')$ and~$\frac{d}{dt}Q(\omega',\omega'')$ are trace-class by Lemma~\ref{lem:trace_class}. 

Integrating the previous expression with respect to $\omega'$ and using Lemma \ref{lem:sokhotski} (with $f(\omega',\omega'')\coloneqq\eta_{\Lambda}(\omega')\tr_{\H_{t}^{\varepsilon}}(Q(\omega',\omega''))$) in order to evaluate the limit as $\delta\to0^{+}$ results in the following distributional equation
    \begin{align*}
        &\slim_{\delta\to0^{+}}\int_{-\infty}^{\infty}d\omega'\eta_{\Lambda}(\omega')\tr_{\H_{t}^{\varepsilon}}\big(F_{\omega'}(H_{\eta})R^{(2)}_{\omega+is\delta}(A_{t})\big)\big|_{s=-1}^{s=1}\\
        &=-\slim_{\delta\to0^{+}}\int_{-\infty}^{\infty}d\omega'\int_{-\infty}^{\infty}d\omega'' \\
        &\qquad \times \eta_{\Lambda}(\omega') \: \tr_{\H_{t}^{\varepsilon}}(Q(\omega',\omega'')) \: \Big[\frac{\partial}{\partial\omega'}\Big(\frac{1}{\omega'-\omega-is\delta}\Big)\Big(\frac{1}{\omega''-\omega-is\delta}\Big)\Big]\Big|^{s=1}_{s=-1}\\
        &=-2\pi i \int_{-\infty}^{\infty} d\omega'\Big( \partial_{\omega}(\eta_{\Lambda}(\omega)\tr_{\H_{t}^{\varepsilon}}(Q(\omega,\omega'))) + 
\partial_{\omega'}(\eta_{\Lambda}(\omega')\tr_{\H_{t}^{\varepsilon}}(Q(\omega',\omega)))\Big)\frac{1}{\omega'-\omega}\:.
    \end{align*}
    Recall that the previous expression only holds in the distributional sense (i.e.\ it is not a pointwise limit). Then, inserting the previous expression in~\eqref{eq:Bt2} and approximating the characteristic function~$g$
by a sequence of test functions~$g_\ell$, we obtain
    \begin{align*}
        B_{t}^{(2)}&=-\frac{1}{2}\int_{-\infty}^\infty d\omega \int_{-\infty}^\infty d\omega' \\
        &\qquad \times \frac{d}{dt}\Big( \partial_{\omega}(\eta_{\Lambda}(\omega)\tr_{\H_{t}^{\varepsilon}}(Q(\omega,\omega'))) + 
\partial_{\omega'}(\eta_{\Lambda}(\omega')\tr_{\H_{t}^{\varepsilon}}(Q(\omega',\omega))) \Big) \frac{g_\ell(\omega') - g_\ell(\omega)}{\omega'-\omega}\\
		&=-\int_{-\infty}^\infty d\omega \int_{-\infty}^\infty d\omega' \frac{d}{dt}\Big( \partial_{\omega}\big(\eta_{\Lambda}(\omega)\tr_{\H_{t}^{\varepsilon}}(Q(\omega,\omega'))\big)\Big) \frac{g_\ell(\omega') - g_\ell(\omega)}{\omega'-\omega}\\
		&=-\int_{-\infty}^\infty d\omega \int_{-\infty}^\infty d\omega'\partial_{\omega}\big(\eta_{\Lambda}(\omega)\tr_{\H_{t}^{\varepsilon}}(\dot{Q}(\omega,\omega'))\big)\frac{g_\ell(\omega') - g_\ell(\omega)}{\omega'-\omega}\:,
    \end{align*} 
    where in the second step we used that the integrand is symmetric if we interchange~$\omega$ and~$\omega'$, so upon re-labeling the integration indices we only have one integral and where~$\dot{Q}(\omega,\omega')\coloneqq\frac{d}{dt}Q(\omega,\omega')$.
Choosing the test function~$g_\ell$ as a sequence which converges pointwise to the characteristic function
and taking the limit gives the result.
The last limit exists in view of the following consideration: The difference of characteristic
functions~$g(\omega')-g(\omega)$ is zero unless~$\omega$ lies in the interval~$[-1/\varepsilon, -m]$
and~$\omega'$ lies outside, or vice versa. Since we may replace the characteristic function near~$-1/\varepsilon$
by a smooth cutoff function, we only need to analyze the boundary~$-m$.
Then the factor~$1/(\omega'-\omega)$ has a pole only at~$\omega=\omega'=-m$.
But in this case, the factors~$F_{\omega}(H_\eta)$ and~$F_{\omega'}(H_\eta)$
and therefore also~$\dot{Q}(\omega,\omega')$ vanish by Lemma~\ref{lem:spectral_measure}.
In this way, the pole at~$\omega=\omega'=-m$ disappears, and the limit is well-defined.
\end{proof}

\section{Discussion}\label{sec:discussion}
Our main result stated in Theorem~\ref{theo:baryo_rate_Minkowski} quantifies the rate of baryogenesis
to first and second order in the deviation from the Dirac dynamics as encoded in the regularizing vector field.
The first order always vanishes. The second order, however, is in general non-zero. This shows that baryogenesis
is a very small, but in general non-vanishing effect. As pointed out in the introduction,
choosing~$u=\partial_t$ no baryogenesis occurs.

The above study of baryogenesis in Minkowski spacetime sets the stage for working out the baryogenesis mechanism in specific, physically relevant cosmological spacetimes.
It is expected that baryogenesis occurred in an early stage of the evolution of the Universe (e.g.\ during the inflationary era) in which metric perturbations of the FLRW spacetimes played an important role. For this reason, we would like to extend our analysis to conformally flat spacetimes and perturbations thereof.
A particular class of promising spacetimes that enter this category are the generalized Robertson-Walker (GRW) spacetimes (a specific type of twisted product spacetime, cf.~\cite{sanchezGRW}). Other interesting scenarios in which we would like to apply our analysis are to the de Sitter and Milne-like spacetimes (see~\cite[Sections~3 and~4]{milnelike} for an introduction to the latter) which also model the early Universe.

\par Secondly, our study naturally prompts a number of interesting mathematical questions.
For example, which properties does the regularizing vector field have?
What are its connections to the flow of the null geodesics? Under which conditions
is the regularizing vector field divergence free? We also point out that all our constructions depend
on the choice of a foliation.
The physical picture is that the regularizing vector field distinguishes a foliation.
How can this picture be made mathematically precise? 
In our above analysis, we assumed that the regularizing vector field coincides with the normal
to the foliation, up to small corrections. Can this picture be extended to a general globally hyperbolic spacetime?
The ultimate goal is not to begin with a foliation, but instead to construct both the foliation
and the regularizing vector together in one construction step. How can this be accomplished?
Analyzing these questions should complete the geometric picture of how baryogenesis arises
in the general globally hyperbolic setting.

\appendix
\section{Kernel identities} \label{sec:computations}
We now explicitly compute the kernels~$F_{\omega}(x,x)$ and~$\left(\partial_{\mu}F_{\omega}(x,y)\right)|_{y=x}$. From expression~\eqref{Pepsdef_omega} it directly follows that:
\begin{align}
        P^{\varepsilon,\omega}(x,x)&=\int \frac{d^{4}k}{(2\pi)^{4}}(\gamma_{\eta j}k^{j}+m)\delta(k^{2}-m^{2})\Theta(-k^{0}+\omega)\Theta(1+\varepsilon k^{0})\nonumber\\
        \left(\partial_{\mu}P^{\varepsilon,\omega}(x,y)\right)\Big|_{y=x}&=-i\int \frac{d^{4}k}{(2\pi)^{4}}k^{\mu}(\gamma_{\eta j}k^{j}+m)\delta(k^{2}-m^2)\Theta(-k^{0}+\omega)\Theta(1+\varepsilon k^{0})\nonumber \:.
    \end{align}
\noindent\underline{1. $F_{\omega}(x,x)$:}\\
We have that:
\begin{align*}
    \frac{d}{d\omega}P^{\varepsilon,\omega}(x,x)&=\int \frac{d^{4}k}{(2\pi)^{4}}(\gamma_{\eta j}k^{j}+m)\delta(k^{2}-m^{2})\delta(k^{0}-\omega)\Theta(1+\varepsilon k^{0})\\
    &=\int \frac{d^{3}k}{(2\pi)^{4}}(\gamma_{\eta0}\omega+\gamma_{\eta\mu}k^{\mu}+m)\delta(\omega^{2}-|\vec{k}|^{2}-m^{2})\Theta(1+\varepsilon \omega)\\
    &=:\frac{1}{(2\pi)^{4}}\big((\gamma_{\eta0}\omega+m)J+\gamma_{\eta\mu}J^{\mu}\big)\Theta(1+\varepsilon\omega)
\end{align*}
where~$\sqrt{k^{\mu}k_{\mu}}=|\vec{k}|$. We compute the integrals~$J, J^{\mu}$ appearing in the previous expression using spherical coordinates ~$(r,\theta, \phi)\in (0,\infty)\times(0,\pi)\times(0,2\pi)$ and introducing the integration variable~$a=r^{2}$:
\begin{align}
J &\coloneqq \int d^{3}k \delta(\omega^{2}-|\vec{k}|^{2}-m^{2}) = 4\pi\int_{0}^{\infty}\delta(r^{2}-(\omega^{2}-m^{2}))dr\nonumber\\
    &\;=2\pi\int_{0}^{\infty}\sqrt{a} \delta(a-(\omega^{2}-m^{2}))da=2\pi\sqrt{\omega^{2}-m^{2}}\\
    J^{\mu} &\coloneqq \int d^{3}k k^{\mu}\delta(\omega^{2}-|\vec{k}|^{2}-m^{2})=0\:,
\end{align}
where~$J^{\mu}$ vanishes by spherical symmetry: i.e.\ writing~$k^{\mu}$ in spherical coordinates, we have that:
    \begin{align*}
        \int_{0}^{2\pi} d\varphi\int_{0}^{\pi} \sin{\theta}k^{\mu}\;d\theta =0
    \end{align*}
Hence, $F_{\omega}(x,x)$ is 
\begin{equation}
    F_{\omega}(x,x)=-2\pi\frac{d}{d\omega}P^{\varepsilon,\omega}(x,x)\gamma_{\eta0}=-\frac{1}{(2\pi)^{2}}\Big((\omega+\gamma_{\eta0}m)\sqrt{\omega^{2}-m^{2}}\Big)\Theta(1+\varepsilon\omega)\:.
\end{equation}
\\\noindent\underline{2. $\left(\partial_{\mu}F_{\omega}(x,y)\right)|_{y=x}$:}\\
In this case it holds that
\begin{align*}
    \frac{d}{d\omega}(\partial_{\mu}P^{\varepsilon,\omega}(x,y))|_{y=x}&=-i\int\frac{d^{4}k}{(2\pi)^{4}}k^{\mu}(\gamma_{\eta j}k^{j}+m)\delta(k^{2}-m^{2})\Theta(1+\varepsilon k^{0})\delta(-k^{0}+\omega)\\
    &=-i\int\frac{d^{3}k}{(2\pi)^{4}} k^{\mu}(\gamma_{\eta0}\omega+\gamma_{\eta\nu}k^{\nu}+m)\delta(\omega^{2}-|\vec{k}|^{2}-m^{2})\Theta(1+\varepsilon \omega)\\
    &=:-\frac{i}{(2\pi)^{4}}\big((\gamma_{\eta0}\omega+m)J^{\mu}+\gamma_{\eta\nu}J^{\mu\nu}\big)\Theta(1+\varepsilon\omega)\:,
\end{align*}
where the integrals~$J^{\mu}$ were already shown to vanish above and it remains to compute the integrals~$J^{\mu\nu}$. We start with the case that~$\mu=\nu$ using again spherical coordinates and introducing the integration variable~$a=r^{2}$
\begin{align}
    J^{\mu\mu} &\coloneqq \int d^{3}k (k^{\mu})^{2}\delta(\omega^{2}-|\vec{k}|^{2}-m^{2})=\frac{4}{3}\pi\int_{0}^{\infty}r^{4}\delta(r^{2}-(\omega^{2}-m^{2}))dr\nonumber\\
    &\;=\frac{2}{3}\pi\int_{0}^{\infty}a^{3/2}\delta(a-(\omega^{2}-m^{2}))da=\frac{2}{3}\pi(\omega^{2}-m^{2})^{3/2}\:.
\end{align}

On the other hand, if~$\mu\neq\nu$, the integrals~$J^{\mu\nu}$ vanish again by spherical symmetry: if~$k^{\mu},k^{\nu}$ are expressed in spherical coordinates it holds that
    \begin{align*}
        &\int_{0}^{2\pi} d\varphi\int_{0}^{\pi} \sin{\theta}k^{\mu}k^{\nu}\;d\theta =0\:.
    \end{align*}
In conclusion, $\left(\partial_{\mu}F_{\omega}(x,y)\right)|_{y=x}$ is
\begin{align}
    \left(\partial_{\mu}F_{\omega}(x,y)\right)|_{y=x}&=-2\pi\frac{d}{d\omega}(\partial_{\mu}P^{\varepsilon,\omega}(x,y))|_{y=x}\gamma_{\eta0}\nonumber\\
    &=\frac{i}{3(2\pi)^{2}}\Big(\gamma_{\eta\mu}\gamma_{\eta0}(\omega^{2}-m^{2})^{3/2}\Big)\Theta(1+\varepsilon\omega)\:.
\end{align}

\Thanks{{{\em{Acknowledgments:}} We would like to thank Claudio F.\ Paganini for insightful discussions.
We are grateful to the referee for helpful comments.
The second author gratefully acknowledges support by the Studienstiftung des deutschen Volkes.

\bibliographystyle{amsplain}
\providecommand{\bysame}{\leavevmode\hbox to3em{\hrulefill}\thinspace}
\providecommand{\MR}{\relax\ifhmode\unskip\space\fi MR }
\providecommand{\MRhref}[2]{%
  \href{http://www.ams.org/mathscinet-getitem?mr=#1}{#2}
}
\providecommand{\href}[2]{#2}

\end{document}